\documentclass[amsthm]{elsart}

\usepackage{yjsco}
\usepackage{natbib}

\usepackage{graphicx}
\usepackage{amssymb}
\usepackage{amsmath}
\usepackage{enumerate}
\usepackage{framed}
\usepackage{colortbl}
\usepackage{pdflscape}
\usepackage{multirow}
\usepackage{color}
\usepackage{bm}

\usepackage[norelsize, ruled,vlined,linesnumbered]{algorithm2e}

\def\R {\ensuremath{\mathbb{R}}}

\def\cont{\mathop{\rm cont}\nolimits}
\def\prim{\mathop{\rm prim}\nolimits}
\def\disc{\mathop{\rm disc}\nolimits}
\def\ldcf{\mathop{\rm ldcf}\nolimits}
\def\coeff{\mathop{\rm coeff}\nolimits}
\def\res{\mathop{\rm res}\nolimits}

\journal{Journal of Symbolic Computation}

\begin{document}

\begin{frontmatter}

\title{Truth Table Invariant \\ Cylindrical Algebraic Decomposition}

\thanks{This work was supported by EPSRC grant EP/J003247/1.}

\author[Bath]{Russell~Bradford},
\ead{R.J.Bradford@bath.ac.uk}
\author[Bath]{James~H.~Davenport},
\ead{J.H.Davenport@bath.ac.uk}
\author[Coventry]{Matthew~England},
\ead{Matthew.England@coventry.ac.uk}
\author[Macquarie]{Scott~McCallum} and
\ead{Scott.McCallum@mq.edu.au}
\author[Bath]{David~Wilson}
\ead{David.John.Wilson@me.com}

\address[Bath]{Department of Computer Science, University of Bath, Bath, BA2 7AY, UK}
\address[Coventry]{School of Computing, Electronics and Maths, Faculty of Engineering, Environment and Computing, Coventry University, Coventry, CV1 5FB, UK}
\address[Macquarie]{Department of Computing, Macquarie University, NSW 2109, Australia}

\begin{abstract}

When using cylindrical algebraic decomposition (CAD) to solve a problem with respect to a set of polynomials, it is likely not the signs of those polynomials that are of paramount importance but rather the truth values of certain quantifier free formulae involving them. This observation motivates our article and definition of a Truth Table Invariant CAD (TTICAD).

In ISSAC 2013 the current authors presented an algorithm that can efficiently and directly construct a TTICAD for a list of formulae in which each has an equational constraint.  This was achieved by generalising McCallum's theory of reduced projection operators.  In this paper we present an extended version of our theory which can be applied to an arbitrary list of formulae, achieving savings if at least one has an equational constraint.  We also explain how the theory of reduced projection operators can allow for further improvements to the lifting phase of CAD algorithms, even in the context of a single equational constraint.  

The algorithm is implemented fully in {\sc Maple} and we present both promising results from experimentation and a complexity analysis showing the benefits of our contributions.  
\end{abstract}


\end{frontmatter}

\newpage

\section{Introduction}
\label{sec:Intro}

A \emph{cylindrical algebraic decomposition} (CAD) is a decomposition of $\R^n$ into cells arranged cylindrically (meaning the projections of any pair of cells are either equal or disjoint) each of which is (semi-)algebraic (describable using polynomial relations).  
CAD is a key tool in real algebraic geometry, offering a method for quantifier elimination in real closed fields.  Applications include
the derivation of optimal numerical schemes \citep{EH14},
parametric optimisation \citep{FPM05}, 
robot motion planning \citep{SS83II}, 
epidemic modelling \citep{BENW06}, 
theorem proving \citep{Paulson2012} 
and programming with complex functions \citep{DBEW12}.

Traditionally CADs are produced \emph{sign-invariant} to a given set of polynomials, (the signs of the polynomials do not vary within each cell).  However, this gives far more information than required for most applications.  Usually a more appropriate object is a \emph{truth-invariant} CAD (the truth of a logical formula does not vary within cells).

In this paper we generalise to define \emph{truth table invariant} CADs (the truth values of a list of quantifier-free formulae do not vary within cells) and give an algorithm to compute these directly.  This can be a tool to efficiently produce a truth-invariant CAD for a parent formula (built from the input list), or indeed the required object for solving a problem involving the input list.  Examples of both such uses are provided following the formal definition in Section \ref{subsec:TTICAD}.
We continue the introduction with some background on CAD, before defining our object of study and introducing some examples to demonstrate our ideas which we will return to throughout the paper.  We then conclude the introduction by clarifying the contributions and plan of this paper.

\subsection{Background on CAD}
\label{subsec:Background}

A \emph{Tarski formula} $F(x_1,\ldots,x_n)$ is a Bool\-ean combination ($\land,\lor,\neg,\rightarrow$) of statements about the signs, ($=0,>0,<0$, but therefore $\ne0,\ge0,\le0$ as well), of certain polynomials $f_i(x_1,\ldots,x_n)$ with integer coefficients.  Such statements may involve the universal or existential quantifiers ($\forall,  \exists$).  We denote by QFF a \emph{quantifier-free Tarski formula}.

Given a quantified Tarski formula
\begin{equation}
\label{eqn:QT}
Q_{k+1}x_{k+1}\ldots Q_nx_n F(x_1,\ldots,x_n)
\end{equation}
(where $Q_i\in\{\forall,\exists\}$ and $F$ is a QFF) the \emph{quantifier elimination problem} is to produce $\psi(x_1,\ldots,x_k)$, an equivalent QFF to \eqref{eqn:QT}.

Collins developed CAD as a tool for quantifier elimination over the reals.  He proposed to decompose $\R^n$ cylindrically such that each cell was sign-invariant for all polynomials $f_i$ used to define $F$. Then $\psi$ would be the disjunction of the defining formulae of those cells $c_i$ in $\R^k$ such that (\ref{eqn:QT}) was true over the whole of $c_i$, which due to sign-invariance is the same as saying that (\ref{eqn:QT}) is true at any one \emph{sample point} of $c_i$.

A complete description of Collins' original algorithm is given by \cite{ACM84I}.
The first phase, \emph{projection}, applies a projection operator repeatedly to a set of polynomials, each time producing another set in one fewer variables.  Together these sets contain the \emph{projection polynomials}.
These are used in the second phase, \emph{lifting}, to build the CAD incrementally.  
First $\R$ is decomposed into cells which are points and intervals corresponding to the real roots of the univariate polynomials.  Then $\R^2$ is decomposed by repeating the process over each cell in $\R$ using the bivariate polynomials at a sample point.  Over each cell there are {\em sections} (where a polynomial vanishes) and {\em sectors} (the regions between) which together form a \emph{stack}.  Taking the union of these stacks gives the CAD of $\R^2$.  This is repeated until a CAD of $\R^n$ is produced.  At each stage the cells are represented by (at least) a sample point and an index: a list of integers corresponding to the ordered roots of the projection polynomials which locates the cell in the CAD.  

To conclude that a CAD produced in this way is sign-invariant we need delineability.  A polynomial is \emph{delineable} in a cell if the portion of its zero set in the cell consists of disjoint sections.  A set of polynomials are \emph{delineable} in a cell if each is delineable and the sections of different polynomials in the cell are either identical or disjoint.  The projection operator used must be defined so that over each cell of a sign-invariant CAD for projection polynomials in $r$ variables (the word \emph{over} meaning we are now talking about an $(r+1)$-dim space) the polynomials in $r+1$ variables are delineable.

The output of this and subsequent CAD algorithms (including the one presented in this paper) depends heavily on the variable ordering.  We usually work with polynomials in $\mathbb{Z}[{\bf x}]=\mathbb{Z}[x_1,\ldots,x_n]$ with the variables, ${\bf x}$, in ascending order (so we first project with respect to $x_n$ and continue to reach  univariate polynomials in $x_1$).  The \emph{main variable} of a polynomial (${\rm mvar}$) is the greatest variable present with respect to the ordering.

CAD has doubly exponential complexity in the number of variables \citep{BD07,DH88}.  There now exist algorithms with better complexity for some CAD applications (see for example \cite{BPR96}) but CAD implementations often remain the best general purpose approach. There have been many developments to the theory since Collin's treatment, including the following:
\begin{itemize}
\item Improvements to the projection operator \citep{Hong1990, McCallum1988, McCallum1998, Brown2001a, HDX14}, reducing the number of projection polynomials computed.
\item Algorithms to identify the adjacency of cells in a CAD \citep{ACM84II, ACM88} and following from this the idea of clustering \citep{Arnon1988} to minimise the lifting.
\item Partial CAD, introduced by \cite{CH91}, where the structure of $F$ is used to lift less of the decomposition of $\R^k$ to $\R^n$, if it is sufficient to deduce $\psi$.
\item The theory of equational constraints, \citep{McCallum1999,McCallum2001,BM05}, also aiming to deduce $\psi$ itself, this time using more efficient projections.
\item The use of certified numerics in the lifting phase to minimise the amount of symbolic computation required \citep{Strzebonski2006, IYAY09}.
\item New approaches which break with the normal projection and lifting model: local projection \citep{Strzebonski2014a}, the building of single CAD cells \citep{Brown2013, JdM12} and CAD via Triangular Decomposition \citep{CMXY09}.  The latter is now used for the CAD command built into \textsc{Maple}, and works by first creating a cylindrical decomposition of complex space.  
\end{itemize}

\subsection{TTICAD}
\label{subsec:TTICAD}

\cite{Brown1998} defined a \emph{truth-invariant CAD} as one for which a formula had invariant truth value on each cell.  Given a QFF, a sign-invariant CAD for the defining polynomials is trivially truth-invariant.  Brown considered the refinement of sign-invariant CADs whilst maintaining truth-invariance, while some of the developments listed above can be viewed as methods to produce truth-invariant CADs directly.
We define a new but related type of CAD, the topic of this paper.

\begin{defn}
Let $\{ \phi_i\}_{i=1}^t$ refer to a list of QFFs.
We say a cylindrical algebraic decomposition $\mathcal{D}$ is a {\em Truth Table Invariant} CAD for the QFFs (TTICAD) if the Boolean value of each $\phi_i$ is constant (either true or false) on each cell of $\mathcal{D}$.
\end{defn}

A sign-invariant CAD for all polynomials occurring in a list of formulae would clearly be a TTICAD for the list.  However, we aim to produce smaller TTICADs for many such lists.  We will achieve this by utilising the presence of equational constraints, a technique first suggested by \cite{Collins1998} with key theory developed by \cite{McCallum1999}.

\begin{defn}
Suppose some quantified formula is given:
\begin{equation*}
\phi^* = (Q_{k+1} x_{k+1})\cdots(Q_n x_n) \phi({\bf x})
\end{equation*}
where the $Q_i$ are quantifiers and $\phi$ is quantifier free.
An equation $f=0$ is an {\em equational constraint} (EC) of $\phi^*$ if $f=0$ is logically implied by $\phi$ (the quantifier-free part of $\phi^*$).
Such a constraint may be either explicit (an atom of the formula) or otherwise implicit.
\end{defn}

In Sections \ref{sec:TTIProj} and \ref{sec:Algorithm} we will describe how TTICADs can be produced efficiently when there are ECs present in the list of formulae.
There are two reasons to use this theory.
\begin{enumerate}

\item \emph{As a tool to build a truth-invariant CAD efficiently:}  If a parent formula $\phi^{*}$ is built from the formulae $\{\phi_i\}$ then any TTICAD for $\{\phi_i\}$ is also truth-invariant for $\phi^{*}$.

We note that for such a formula a TTICAD may need to contain more cells than a truth-invariant CAD.  For example, consider a cell in a truth-invariant CAD for
$\phi^{*} = \phi_1 \lor \phi_2$ within which $\phi_1$ is always true.  If $\phi_2$ changed truth value in such a cell then it would need to be split in order to achieve a TTICAD, but this is unnecessary for a truth-invariant CAD of $\phi^*$.

Nevertheless, we find that our TTICAD theory is often able to produce smaller truth-invariant CADs than any other available approach.  We demonstrate the savings offered via worked examples introduced in the next subsection.

\item \emph{When given a problem for which truth table invariance is required:}  That is, a problem for which the list of formulae are not derived from a larger parent formula and thus a truth-invariant CAD for their disjunction may not suffice.

For example, decomposing complex space according to a set of branch cuts for the purpose of algebraic simplification \citep{BD02, PBD10}.  Here the idea is to represent each branch cut as a semi-algebraic set to give input admissible to CAD, (recent progress on this has been described by \cite{EBDW13}).  Then a TTICAD for the list of formulae these sets define provides the necessary decomposition.  Example \ref{ex:Kahan} is from this class.

\end{enumerate}

\subsection{Worked examples}
\label{subsec:WE1}

To demonstrate our ideas we will provide details for two worked examples.  Assume we have the variable ordering $x \prec y$ (meaning 1-dimensional CADs are with respect to $x$) and consider the following polynomials, graphed in Figure \ref{fig:WE1}.
\begin{align*}
f_1 := x^2+y^2-1 \qquad \qquad \qquad & g_1 := xy - \tfrac{1}{4} \\
f_2 := (x-4)^2+(y-1)^2-1  \quad & g_2 := (x-4)(y-1) - \tfrac{1}{4}
\end{align*}
Suppose we wish to find the regions of $\R{}^2$ where the following formula is true:
\begin{equation}
\label{eqn:ExPhi}
\Phi:= \left(f_1 = 0 \land g_1 < 0 \right)\lor \left( f_2 = 0 \land g_2 < 0   \right).
\end{equation}
Both \textsc{Qepcad} \citep{Brown2003a} and \textsc{Maple} 16 \citep{CMXY09} produce a sign-invariant CAD for the polynomials with 317 cells.  Then by testing the sample point from each region we can systematically identify where the formula is true.

\begin{figure}[t]
\caption{The polynomials from the worked examples of Section~\ref{subsec:WE1}.  The solid curves are $f_1$ and $g_1$ while the dashed curves are $f_2$ and $g_2$.}
\label{fig:WE1}
\begin{center}
\includegraphics[scale=0.4]{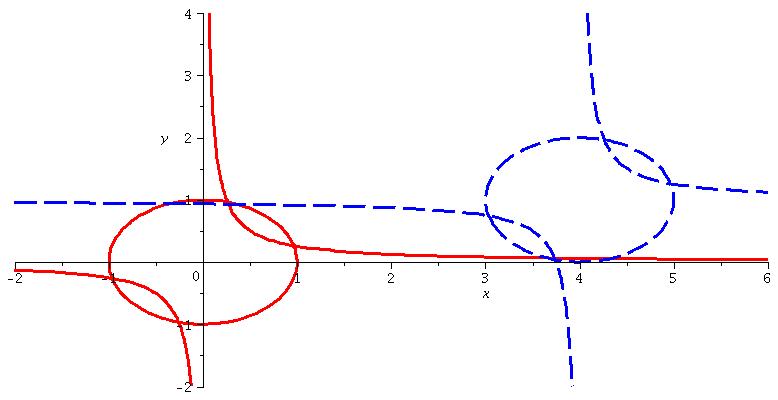}
\end{center}
\end{figure}

At first glance it seems that the theory of ECs is not applicable to $\Phi$ as neither $f_1 = 0$ nor $f_2 = 0$ is logically implied by $\Phi$.  However, while there is no explicit EC we can observe that $f_1f_2 = 0$ is an {\em implicit} constraint of $\Phi$.  Using \textsc{Qepcad} with this declared (an implementation of \citep{McCallum1999}) gives a CAD with 249 cells.  Later, in Section \ref{subsec:WE2} we demonstrate how a TTICAD with 105 cells can be produced.

We also consider the related problem of identifying where
\begin{equation}
\label{eqn:ExPsi}
\Psi:= \left(f_1 = 0 \land g_1 < 0 \right) \lor \left( f_2 > 0 \land g_2 < 0   \right)
\end{equation}
is true.  As above, we could use a sign-invariant CAD with 317 cells, but this time there is no implicit EC.  In Section \ref{subsec:WE2} we produce a TTICAD with 183 cells.

\subsection{Contributions and plan of the paper}
\label{subsec:Plan}

We review the projection operators of \cite{McCallum1998, McCallum1999} in Section \ref{sec:ExistingProj}.  The former produces sign-invariant CADs\footnote{Actually order-invariant CADs (see Definition \ref{def:OI}).}and the latter CADs truth-invariant for a formula with an EC.  The review is necessary since we use some of this theory to verify our new algorithm.  It also allows us to compare our new contribution to these existing approaches.  For this purpose we provide new complexity analyses of these existing theories in Section \ref{subsec:CA1}.

Sections \ref{sec:TTIProj} and \ref{sec:Algorithm} present our new TTICAD projection operator and verified algorithm.  They follow Sections 2 and 3 of our ISSAC 2013 paper \citep{BDEMW13}, but instead of requiring all QFFs to have an EC the theory here is applicable to all QFFs (producing savings so long as one has an EC).  The strengthening of the theory means that a TTICAD can now be produced for $\Psi$ in Section \ref{subsec:WE1} as well as $\Phi$. 
This extension is important since it means TTICAD theory now applied to cases where there can be no overall implicit EC for a parent formula.  In these cases the existing theory of ECs is not applicable and so the comparative benefits offered by TTICAD are even higher.

In Section \ref{sec:ImprovedLifting} we discuss how the theory of reduced projection operators also allows for improvements in the lifting phase.  This is true for the existing theory also but the discovery was only made during the development of TTICAD.  
In Section \ref{sec:CA} we present a complexity analysis of our new contributions from Sections \ref{sec:TTIProj} $-$ \ref{sec:ImprovedLifting}, demonstrating their benefit over the existing theory from Section \ref{sec:ExistingProj}.
We have implemented the new ideas in a \textsc{Maple} package, discussed in Section \ref{sec:Implementation}.   In particular, Section \ref{sec:Formulation} summarises \citep{BDEW13} on the choices required when using TTICAD and heuristics to help.
Experimental results for our implementation (extending those in our ISSAC 2013 paper) are given in Section \ref{sec:Experiment}, before we finish in Section \ref{sec:Conclusion} with conclusions and future work.

\textbf{Data access statement:} Data directly supporing this paper (code, \textsc{Maple} and \textsc{Qepcad} input) is openly available from \texttt{http://dx.doi.org/10.15125/BATH-00076}.

\section{Existing CAD projection operators}
\label{sec:ExistingProj}

\subsection{Review: Sign-invariant CAD}
\label{subsec:SI}

Throughout the paper we let $\cont, \prim, \disc, \coeff$ and $\ldcf$ denote the content, primitive part, discriminant, coefficients and leading coefficient of polynomials respectively (in each case taken with respect to a given main variable).  Similarly, we let $\res$ denote the resultant of a pair of polynomials.  When applied to a set of polynomials we interpret these as producing sets of polynomials, so for example
\[
\res(A)=\left\{\res(f_i,f_j) \, | \, f_i \in A, f_j \in A, f_j \neq f_i \right\}.
\]

The first improvements to Collins original projection operator were given by \cite{McCallum1988} and \cite{Hong1990}.  They were both subsets of Collins operator, meaning fewer projection polynomials, fewer cells in the CADs produced and quicker computation time.  McCallum's is actually a strict subset of Hong's, however, it cannot be guaranteed correct (incorrectness is detected in the lifting process) for a certain class of (statistically rare) input polynomials, where Hong's can.  

Additional improvements have been suggested by \cite{Brown2001a} and \cite{Lazard1994}.  The former required changes to the lifting phase while the latter had a flawed proof of validity (with current unpublished work suggesting it can still be safely used in many cases).  
In this paper we will focus on McCallum's operators, noting that the alternatives could likely be extended to TTICAD theories too if desired.
McCallum's theory is based around the following condition, which implies sign-invariance.
\begin{defn}
\label{def:OI}
A CAD is \emph{order-invariant} with respect to a set of polynomials if each polynomial has constant order of vanishing within each cell.
\end{defn}

Recall that a set $A \subset \mathbb{Z}[{\bf x}]$ is an \emph{irreducible basis} if the elements of $A$ are of positive degree in the main variable, irreducible and pairwise relatively prime.  Let $A$ be a set of polynomials and $B$ an irreducible basis of the primitive part of $A$. Then
\begin{equation}
\label{eq:P}
P(A):=\cont(A) \cup \coeff(B) \cup \disc(B) \cup \res(B)
\end{equation}
defines the operator of \cite{McCallum1988}.  We can assume some trivial simplifications such as the removal of constants and exclusion of entries identical to a previous one (up to constant multiple).
The main theorem underlying the use of $P$ follows.
\begin{thm}[\cite{McCallum1998}]
\label{thm:McC1}
Let $A$ be an irreducible basis in $\mathbb{Z}[{\bf x}]$ and let $S$ be a connected submanifold of $\mathbb{R}^{n-1}$. Suppose each element of $P(A)$  is order-invariant in $S$.

Then each element of $A$ either vanishes identically on $S$ or is analytic delineable on $S$, (a slight variant on traditional delineability, see \citep{McCallum1998}). Further, the sections of $A$ not identically vanishing are pairwise disjoint, and each element of $A$ not identically vanishing is order-invariant in such sections.
\end{thm}
Theorem \ref{thm:McC1} means that we can use $P$ in place of Collins' projection operator to produce sign-invariant CADs so long as none of the projection polynomials with main variable $x_k$ vanishes on a cell of the CAD of $\R^{k-1}$; a condition that can be checked when lifting.  Input with this property is known as \emph{well-oriented}.  Note that although McCallum's operator produces order-invariant CADs, a stronger property than sign-invariance, it is actually more efficient that the pre-existing sign-invariant operators. We examine the complexity of CAD using this operator in Section \ref{subsec:CA1}.

\subsection{Review: CAD invariant with respect to an equational constraint}
\label{subsec:EC}

The main result underlying CAD simplification in the presence of an EC follows.
\begin{thm}[\cite{McCallum1999}]
\label{thm:McC2}
Let $f({\bf x}), g({\bf x})$ be integral polynomials with positive degree in $x_n$,
let $r(x_1,\ldots,x_{n-1})$ be their resultant, and suppose $r \neq 0$.
Let $S$ be a connected subset of $\mathbb{R}^{n-1}$
such that $f$ is delineable on $S$ and $r$ is order-invariant in $S$.

Then $g$ is {\em sign-invariant} in every section of $f$ over $S$.
\end{thm}

\begin{figure}[ht]
\caption{Graphical representation of Theorem \ref{thm:McC2}.}
\label{fig:Theorem}
\begin{center}
\includegraphics[scale=0.5]{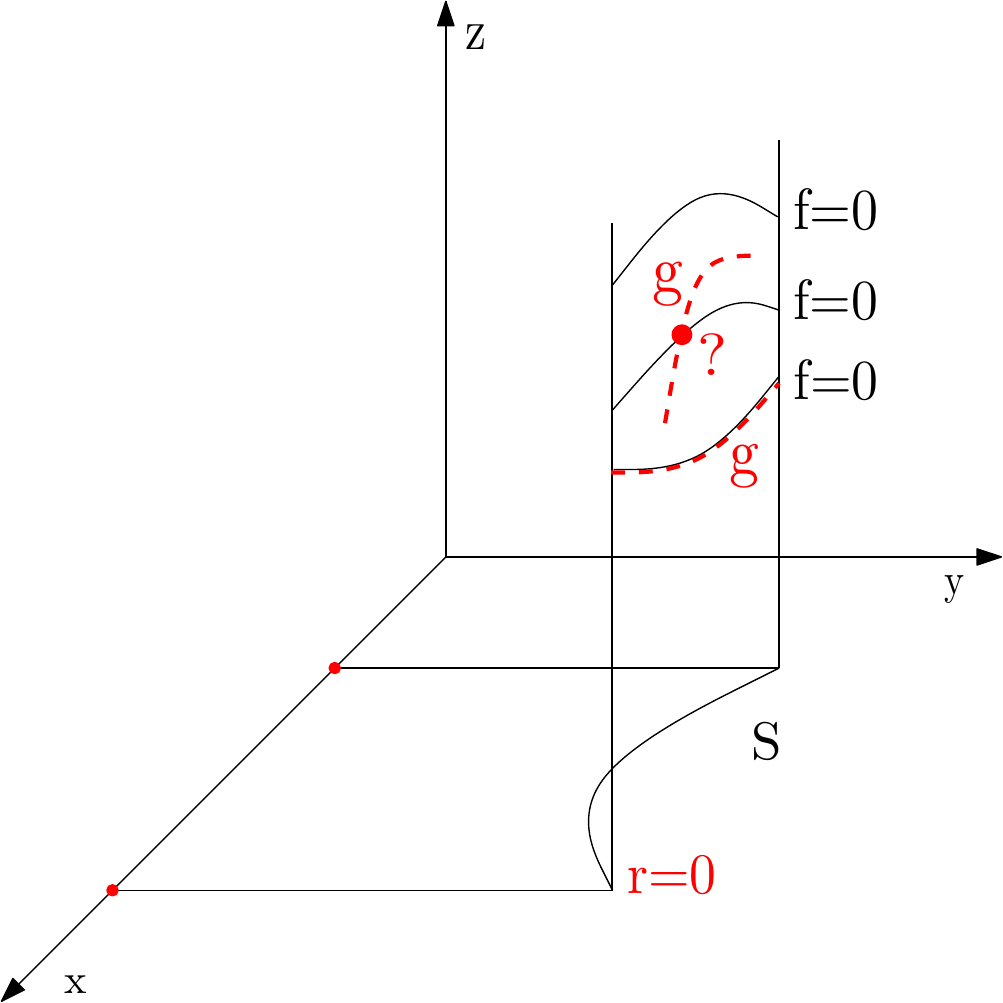}
\end{center}
\end{figure}

Figure \ref{fig:Theorem} gives a graphical representation of the question answered by Theorem \ref{thm:McC2}.
Here we consider polynomials $f(x,y,z)$ and $g(x,y,z)$ of positive degree in $z$ whose resultant $r$
is non-zero, and a connected subset $S \subset \mathbb{R}^2$ in which $r$ is order-invariant.  We further suppose that $f$ is delineable on $S$ (noting that Theorem \ref{thm:McC1} with $n=3$ and $A = \{f\}$ provides sufficient conditions for this).
We ask whether $g$ is sign-invariant in the sections of $f$ over $S$.  Theorem \ref{thm:McC2} answers this question affirmatively:  the real variety of $g$ either aligns with a given section of $f$ exactly (as for the bottom section of $f$ in Figure \ref{fig:Theorem}), or has no intersection with such a section (as for the top).
The situation at the middle section of $f$ cannot happen.

Theorem \ref{thm:McC2} thus suggests a reduction of the projection operator $P$ relative to an EC $f = 0$: take only $P(f)$ together with the resultants of $f$ with the non-ECs.  
Let $A$ be a set of polynomials, $E \subset A$ contain only the polynomial defining the EC, $F$ be a square free basis of $A$, and $B$ be the subset of $F$ which is a square-free basis for $E$.  The operator
\begin{equation}
\label{eq:ECProj}
P_{E}(A) := \cont(A) \cup P(F) \cup \{ {\rm res}_{x_n}(f,g) \mid f \in F, g \in B \setminus F \}
\end{equation}
was presented by \cite{McCallum1999} along with an algorithm to produce a CAD truth-invariant for the EC and sign-invariant for the other polynomials when the EC was satisfied.  It worked by applying first $P_E(A)$ and then building an order-invariant CAD of $\mathbb{R}^{n-1}$ using $P$.  We call such CADs  \emph{invariant with respect to an equational constraint}.  Note that as with \cite{McCallum1999} the algorithm only works for input satisfying a well-orientedness condition.  Full details of the verification are given by \cite{McCallum1999} and a complexity analysis is given in the next subsection.

\subsection{New complexity analyses}
\label{subsec:CA1}

We provide complexity analyses of the algorithms from \cite{McCallum1998, McCallum1999} for comparison with our new contributions later.  An analysis for the latter has not been published before, while the analysis for the former differs substantially from the one in \citep{McCallum1985}:  instead of focusing on computation time, we examine the number of cells in the CAD of $\mathbb{R}^n$ produced: the \emph{cell count}.  We compare the dominant terms in a cell count bound for each algorithm studied.  This focus avoids calculations with less relevant parameters, identical for all the algorithms.  We note that all CAD experimentation shows a strong correlation between the number of cells produced and the computation time.  

Our key parameters are the number of variables $n$, the number of polynomials $m$ and their maximum degree $d$ (in any one variable).  Note that these are all restricted to positive integer values.  We make much use of the following concepts.

\begin{defn}
\label{def:cd}
Consider a set of polynomials $p_j$.  The \textbf{combined degree} of the set is the maximum degree (taken with respect to each variable) of the product of all the polynomials in the set:
$
\textstyle \max_{i} \left( \deg_{x_i}\left( \prod_j p_j \right)\right).
$
\end{defn}
So for example, the set $\{x^2+1, x^2+y^3\}$ has combined degree $4$ (since the product has degree $4$ in $x$ and degree $3$ in $y$).

\begin{defn}[\cite{McCallum1985}]
\label{def:md}
A set of polynomials has the $\bm{(m,d)}$\textbf{-property} if it can be partitioned into $m$ sets, such that each set has maximum combined degree $d$.
\end{defn}
So for example, the set of polynomials 
$
\{xy^3-x, x^4-xy, x^4-y^4+1\}
$
has combined degree $9$ and thus the $(1,9)$-property.   However, by partitioning it into three sets of one polynomial each, it also has the $(3,4)$-property.  Partitioning into 2 sets will show it to have the $(2, 5)$, $(2,7)$ and $(2,8)$-properties also.  

The following result follows simply from the definitions.  
\begin{prop}
\label{Prop:MD1}
If $A$ has the $(m,d)$-property then so does any squarefree basis of $A$.
\end{prop}
This contrasts with the facts that taking a square-free basis may not reduce the combined degree, but may cause exponential blow-up in the number of polynomials.
\begin{prop}
\label{Prop:MD2}
Suppose a set has the $(m,d)$-property.  Then, by taking the union of groups of $\ell$ sets from the partition, it also has the $\left( \left\lceil \tfrac{m}{\ell} \right\rceil, \ell d \right)$-property.
\end{prop}
Note that in the case $\ell=2$ we have 
$\left\lceil \tfrac{m}{2} \right\rceil = \left\lfloor \tfrac{m+1}{2} \right\rfloor$.


\begin{exmp}
\label{ex:md}
Let $S = \{ x^2y^4-x^3, x^2y^4+x^3 \}$ be a set of polynomials.  Then $S$ has the $(2,4)$ and $(1,8)$-properties.  A squarefree basis of $S$ is given by $S' = \{x^2, y^4-x, y^4+x\}$ which has the $(3,4)$ and $(1,8)$-properties.  

Proposition \ref{Prop:MD2} states that $S'$ must also have the $(2,8)$-property, which can be checked by partitioning $S'$ so that $x^2$ is in a set of its own.  However, from Proposition \ref{Prop:MD1} we also know that $S'$ must have the $(2,4)$-property, which is obtained from either of the other partitions into two sets.  

$S'$ demonstrates the strength of the $(m,d)$-property.  The trivial partition into sets of one polynomial is equivalent to the simple approach of just tracking the number of polynomials and maximum degree.  In this example such an approach would lead us to 3 polynomials of degree 4, contributing a possible 12 real roots.  However, by using more sophisticated partitions we replace this by 2 sets, for each of which the product of polynomial entries has degree 4, and so at most 8 real roots contributed.   
\end{exmp}

Though not used in this paper, we note an advantage of the $(m,d)$-property over the $(1,md)$-property is a better bound on root separation: any two roots require $O(2d)$ subdivisions to isolate, rather than the $O(md)$ implied by considering the product of all polynomials. 

We also recall the following classic identities for polynomials $f,g,h$:
\begin{align}
\res(fg,h) &= \res(f,h)\res(g,h); \label{eqn:ResProp1} \\
\disc(fg)  &= \disc(f)\disc(g)\res(f,g)^2; \label{eqn:ResProp2} \\
\disc(f)   &=(-1)^{\frac{1}{2}d(d-1)}\tfrac{1}{a_d}\res(f,f') \label{eqn:ResProp3} 
\end{align}
where $d$ is the degree of $f$, $f'$ its derivative and $a_d$ its leading coefficient (all taken with respect to the given main variable).

\begin{lem}
\label{L:SINew}
Suppose $A$ is a set of polynomials in $n$ variables with the $(m,d)$ property.  Then $P(A)$ has the $(M,2d^2)$ property with 
\begin{equation}
\label{eq:SI-M}
M = \left \lfloor \frac{(m+1)^2}2 \right\rfloor.
\end{equation}
\end{lem}

\begin{proof}
Partition $A$ as $S_1\cup\cdots\cup S_m$ according to its $(m,d)$-property.  Let $B$ be a square-free basis for $\prim(A)$, $T_1$ the set of elements of $B$ which divide some element of $S_1$, and $T_i$ be those elements of $B$ which divide some element of $S_i$ but which have not already occurred in some $T_j: j<i$. 
\begin{enumerate}
\item We first claim that each set
\begin{equation}
\label{eq:T}
\cont(S_i)\cup\ldcf(T_i)\cup\disc(T_i)\cup\res(T_i)
\end{equation}
for $i = 1, \dots m$ has the $(1,2d^2)$ property. Let $c$ be the product of the elements of $\cont(S_i)$, $T_i=\{F_1,\ldots,F_{\mathfrak{t}}\}$ for some $\mathfrak{t}$ and $F:=cF_1,\ldots F_{\mathfrak{t}}$.  Then $F$ divides the product of the elements of $S_i$ and so has degree at most $d$. Thus $\res(F,F')$ must have degree at most $2d^2$ because it is the determinant of a $(2d-1 \times 2d-1)$ matrix in which each element has degree at most $d$. Then by (\ref{eqn:ResProp3}) and repeated application of (\ref{eqn:ResProp1}) and (\ref{eqn:ResProp2}) we see $\res(F,F')$ is a (non-trivial) power of $c$ multiplied by
\[
\textstyle \prod_{j=1}^{t} \ldcf(F_j) \prod_{j=1}^{t} \disc(F_j) \prod_{j<k}^{t} \res(F_j,F_k)^2.
\]
Since this includes all the elements of (\ref{eq:T}) the claim is proved.
\item We are still missing from $P(A)$ the $\res(f,g)$ where $f \in T_i, g \in T_j$ and $i \ne j$. For fixed $i,j$ consider $\res\left(\prod_{f\in T_i}f,\prod_{g\in T_j}g\right)$, which by (\ref{eqn:ResProp1}) is the product of the missing resultants.  This is the resultant of two polynomials of degree at most $d$ and hence will have degree at most $2d^2$.  Thus for fixed $i,j$ the set of missing resultants has the $(1,2d^2)$-property, and so the union of all such sets the $\left( \tfrac{1}{2}m(m-1), 2d^2 \right)$-property.
\item We are now missing from $P(A)$ only the non-leading coefficients of $B$.
The polynomials in the set $T_i$ have degree at most $d$ when multiplied together, and so, separately or together, have at most $d$ \emph{non-leading} coefficients, each of which has degree at most $d$. Hence this set of \emph{non-leading} coefficients has the $(1,d^2)$ property.  This is the case for $i$ from $1$ to $m$ and thus together the non-leading coefficients of $B$ have the $(m,d^2)$-property.  We can then pair up these sets to get a partition with the  $(\lceil m/2\rceil, 2d^2)$-property (Proposition \ref{Prop:MD2}).
\end{enumerate}
Hence $P(A)$ can be partitioned into 
\[
m + \frac{m(m-1)}2 + \left\lceil \frac{m}{2} \right\rceil 
= \frac{m(m+1)}2 + \left\lfloor \frac{m+1}{2} \right\rfloor
= \left \lfloor \frac{(m+1)^2}2 \right\rfloor
\]
sets (where the final equality follows from $m(m+1)$ always being even) each with combined degree $2d^2$.
\end{proof}

This concerns a single projection, and we must apply it recursively to consider the full set of projection polynomials.  Weakening the bound as in the following allows for a closed form solution.

\begin{cor}
\label{cor:SI2}
If $A$ is a set of polynomials with the $(m,d)$ property where $m>1$, then $P(A)$ has the $(m^2,2d^2)$-property.
\end{cor}

\begin{rem}
\label{rem:SIComplexity}
\begin{enumerate}
\item Note that if $A$ has the $(1,d)$-property then $P(A)$ has the $(2,2d^2)$ property and hence the need for $m>1$ to apply Corollary \ref{cor:SI2}.  
As our paper continues we present new theory that applies to the first projection only.  Hence for a fair and accurate complexity comparison we will use Lemma \ref{L:SINew} for the first projection and then Corollary \ref{cor:SI2} for subsequent ones, (applicable since even if we start with $m=1$ polynomial for the first projection, we can assume $m \geq 2$ thereafter). 
\item The analysis so far resembles Section 6.1 of \cite{McCallum1985}.  However, that thesis leads us to the $(m^2d,2d^2)$-property in place of Corollary \ref{cor:SI2}.  The extra dependency on $d$ was avoided by an improved analysis in the proof of Lemma \ref{L:SINew} part (3). 
\end{enumerate}
\end{rem}

We consider the growth in projection polynomials and their degree when using the operator $P$ in Table \ref{tab:GeneralProjection}.  Here the column headings refer not to the number of polynomials and their degree, but to the number of sets and their combined degree when applying Definition \ref{def:md}.  We start with $m$ polynomials of degree $d$ and after one projection have a set with the $(M, 2d^2)$ property, using $M$ from Lemma \ref{L:SINew}.  We then use Corollary \ref{cor:SI2} to model the growth in subsequent projections, and a simple induction to fill in the table.

\begin{table}[t]
\caption{Expression growth for CAD projection where: after the first projection we have polynomials with the ($M, 2d^2$)-property and thereafter we measure growth using Corollary \ref{cor:SI2}.  The value of $M$ could be (\ref{eq:SI-M}), (\ref{eq:EC-M}), (\ref{eq:TTI-M}) , (\ref{eq:ECImplicit-M}) or (\ref{eq:TTIGeneral-M}) depending on which projection scheme we are analysing.}  \label{tab:GeneralProjection}
\begin{center}
\begin{tabular}{cccc}
Variables & Number  & Degree    & Product            \\
\hline
$n$       & $m$     & $d$       & $md$               \\
$n-1$     & $M$     & $2d^2$    & $2Md^2$            \\
$n-2$     & $M^2$   & $8d^4$    & $2^3M^2d^4$        \\
$n-3$     & $M^4$   & $128d^8$  & $2^{7}M^4d^8$      \\
\vdots    & \vdots  & \vdots    & \vdots             \\
$n-r$     & $M^{2^{r-1}}$ & $2^{2^r-1}d^{2^r}$ & $2^{2^r-1}d^{2^r}M^{2^{r-1}}$ 
\\
\vdots & \vdots & \vdots & \vdots 
\\
1      & $M^{2^{n-2}}$ & $2^{2^{n-1}-1}d^{2^{n-1}}$ & $2^{2^{n-1}-1}d^{2^{n-1}}M^{2^{n-2}}$
\\
\hline
Product & $M^{2^{n-1}-1}m$ & $2^{2^{n}-1-n}d^{2^{n}-1}$ & $2^{2^{n}-n-1}d^{2^n-1}M^{2^{n-1}-1}m$
\end{tabular}
\end{center}
\end{table}

The size of the CAD produced depends on the number of real roots of the projection polynomials.  We can hence bound the number of real roots in a set of polynomials with the $(m,d)$-property with $md$ (in practice many of them will be strictly complex).
We can therefore bound the number of real roots of the univariate projection polynomials  by the product of the two entries in the row of Table \ref{tab:GeneralProjection} for 1 variable.  The number of cells in the CAD of $\R^1$ is bounded by twice this plus 1. Similarly, the total number of cells in the CAD of $\R^n$ is bounded by the product of $2K+1$ where $K$ varies through the Product column of Table \ref{tab:GeneralProjection}, i.e. by
\[
(2Md + 1) \prod_{r=1}^{n-1} \left[ 2 \left( 2^{2^r-1}d^{2^r}M^{2^{r-1}} \right) + 1 \right].
\]
Omitting the $+1$ will leave us with the dominant term of the bound, which can be calculated explicitly as
\begin{align}
&\qquad 2^{2^{n}-1}d^{2^{n}-1}M^{2^{n-1}-1}m, \label{bound:All} \\
&\leq 2^{2^{n}-1}d^{2^{n}-1}\left( \tfrac{1}{2}(m+1)^2 \right)^{2^{n-1}-1}m   
= 2^{2^{n-1}}d^{2^{n}-1}(m+1)^{2^n-2}m.    \label{bound:SINew}
\end{align}
where the inequality was introduced by omitting the floor function in (\ref{eq:SI-M}).  This may be compared with the bound in Theorem 6.1.5 of \cite{McCallum1985}, with the main differences explained by Remark \ref{rem:SIComplexity}(2).

We now turn our focus to CAD invariant with respect to an EC.  Recall that we use operator $P_E(A)$ for the first projection only and $P(A)$ thereafter.  Hence we use Corollary \ref{cor:SI2} for the bulk of the analysis, and the next lemma when considering the first projection.
 
\begin{lem}
\label{L:EC}
Suppose $A$ is a set of $m$ polynomials in $n$ variables each with maximum degree d, and that $E \subseteq A$ contains a single polynomial.  
Then the reduced projection $P_E(A)$ has the $(M,2d^2)$-property with 
\begin{equation}
\label{eq:EC-M}
M = \left\lfloor \tfrac{1}{2}(3m+1) \right\rfloor.
\end{equation}
\end{lem}
\begin{proof} 
Since $E$ contains a single polynomial its squarefree basis $F$ has the $(1,d)$-property.  
\begin{enumerate}
\item The contents, leading coefficients and discriminants from $F$ form a set $R_1$ with combined degree $2d^2$ (see proof of Lemma \ref{L:SINew} step 1) and the other coefficients a set $R_2$ with combined degree $d^2$ (see proof of Lemma \ref{L:SINew} step 3).  

\item The set of remaining contents $R_3 = \cont(A) \setminus \cont(E)$ has the $(m-1,d)$-property and thus trivially, the $(m-1,d^2)$-property.  Then $R_2 \cup R_3$ has the $(m,d^2)$-property and thus also the $\left( \lceil \tfrac{m}{2} \rceil, 2d^2 \right)$-property (Proposition \ref{Prop:MD2}).  

\item It remains to consider the final set of resultants in (\ref{eq:ECProj}).  Following the approach from the proof of Lemma \ref{L:SINew} step 2, we conclude that for each of $m-1$ polynomials in $A \setminus E$ there contributes a set with the $(1,2d^2)$-property.  So together they form a set $R_4$ with the $(m-1, 2d^2)$-property.
\end{enumerate}
Hence $P_E(A)$ is contained in $R_1 \cup ( R_2 \cup R_3 ) \cup R_4$ which may be partitioned into
\[
1 + \left\lceil \tfrac{m}{2} \right\rceil + (m-1) 
= \left\lfloor \tfrac{1}{2}(m+1) \right\rfloor + m
= \left\lfloor \tfrac{1}{2}(3m+1) \right\rfloor
\]
sets of combined degree $2d^2$.
\end{proof}

%
%

We can use Table \ref{tab:GeneralProjection} to model the growth in projection polynomials for the algorithm in \citep{McCallum1999} as well, since the only difference will be the number of polynomials produced by the first projection, and thus the value of $M$. 
Hence the dominant term in the bound on the total number of cells is given again by (\ref{bound:All}), which in this case becomes (upon omitting the floor)
\begin{align}
&\qquad 2^{2^{n}-1}d^{2^{n}-1}( \tfrac{1}{2}(3m+1) )^{2^{n-1}-1}m 
= 2^{2^{n-1}}d^{2^{n}-1}(3m+1)^{2^{n-1}-1}m.  \label{bound:EC}
\end{align}

Since $P_E(A)$ is a subset of $P(A)$ a CAD invariant with respect to an EC should certainly be simpler than a sign-invariant CAD for the polynomials involved.  Indeed, comparing the different values of $M$ we see that 
\[
\tfrac{1}{2}(m+1)^2 > \tfrac{1}{2}(3m+1) \qquad \mbox{(strictly so for } m>1 \mbox{).}
\]
Comparing the dominant terms in the cell count bounds, (\ref{bound:EC}) and (\ref{bound:SINew}), 
we see the main effect is a decrease in one of the double exponents by $1$.
 
\section{A projection operator for TTICAD}
\label{sec:TTIProj}

\subsection{New projection operator}
\label{subsec:TTIProjOp}

In \citep{McCallum1999} the central concept is the reduced projection of a set of polynomials $A$ relative to a subset $E$ (defining the EC).  The full projection operator is applied to $E$ and then supplemented by the resultants of polynomials in $E$ with those in $E \setminus A$, since the latter group only effect the truth of the formula when they share a root with the former. 
We extend this idea to define a projection for a list of sets of polynomials (derived from a list of formulae), some of which may have subsets (derived from ECs).  

For simplicity in \citep{McCallum1999} the concept is first defined for the case when $A$ is an irreducible basis.  We emulate this approach, generalising for other cases by considering contents and irreducible factors of positive degree when verifying the algorithm in Section \ref{sec:Algorithm}.
So let $\mathcal{A} = \{ A_i\}_{i=1}^t$ be a list of irreducible bases $A_i$ and let $\mathcal{E} = \{ E_i \}_{i=1}^t$ be a list of 
subsets $E_i \subseteq A_i$.
Put $A = \bigcup_{i=1}^t A_i$ and $E = \bigcup_{i=1}^t E_i$.  Note that we use the convention of uppercase Roman letters for sets of polynomials and calligraphic letters for lists of these.

\begin{defn}
\label{def:TTIProj}
With the notation above the {\em reduced projection of $\mathcal{A}$ with respect to $\mathcal{E}$}
is
\begin{equation}
\label{eqn:TTIProj-S}
P_{\mathcal{E}}(\mathcal{A}) := \textstyle{\bigcup_{i=1}^t} P_{E_i}(A_i)
\cup {\rm RES}^{\times} (\mathcal{E})
\end{equation}
where ${\rm RES}^{\times} (\mathcal{E}) $ is the cross resultant set
\begin{align}
{\rm RES}^{\times} (\mathcal{E})
&= \{ {\rm res}_{x_n}(f,\hat{f}) \mid \exists \, i,j \,\, \mbox{such that }
f \in E_i, \hat{f} \in E_j, i<j, f \neq \hat{f}  \} 
\label{eqn:RESX}
\end{align}
and 
\begin{align*}
P_{E}(A) = P(E) \cup \left\{ {\rm res}_{x_n}(f,g) \mid f\in E, g \in A, g \notin E \right\}, 
\\
P(A) = \{ {\rm coeffs}(f), {\rm disc}(f), {\rm res}_{x_n}(f,g) \, | \, f,g \in A, f \neq g \}. 
\end{align*}
\end{defn}


\begin{thm}
\label{thm:Main}
Let $S$ be a connected submanifold of $\mathbb{R}^{n-1}$. Suppose each element of $P_{\mathcal{E}}(\mathcal{A})$ is order invariant in $S$. Then each $f \in E$ either vanishes identically on $S$ or is analytically delineable on $S$; the sections over $S$ of the $f \in E$ which do not vanish identically are pairwise disjoint; and each element $f \in E$ which does not vanish identically is order-invariant in such sections.

\emph{Moreover}, for each $i$, in $1 \leq i \leq t$ every $g \in A_i \setminus E_i$ is sign-invariant in each section over $S$ of every $f \in E_i$ which does not vanish identically.
\end{thm}

\begin{proof}
The crucial observation for the first part is that
$P(E) \subseteq P_{\mathcal{E}}(\mathcal{A})$.
To see this, recall equation \eqref{eqn:TTIProj-S} and note that we can write
\begin{equation*}
P(E) = {\textstyle \bigcup_i }P(E_i) \cup {\rm RES}^{\times}(\mathcal{E}).
\end{equation*}
We can therefore apply Theorem \ref{thm:McC1} to the set $E$ and obtain the first three conclusions immediately, leaving only the final conclusion to prove. 

Let $i$ be in the range $1 \leq i \leq t$, let $g \in A_i \setminus E_i$ and let $f \in E_i$. Suppose that $f$ does not vanish identically on $S$.
Now ${\rm res}_{x_n}(f,g) \in P_{\mathcal{E}}(\mathcal{A})$, and so is order-invariant in $S$ by hypothesis. Further, we already concluded that $f$ is delineable. Therefore by Theorem \ref{thm:McC2}, $g$ is sign-invariant in each section of $f$ over $S$.
\end{proof}

Theorem \ref{thm:Main} is the key tool for the verification of our TTICAD algorithm in Section \ref{sec:Algorithm}.  It allows us to conclude the output is correct so long as no $f \in E$
vanishes identically on the lower dimensional manifold, $S$.  A polynomial $f$ in $r$ variables that vanishes identically at a point  $\alpha \in \R{}^{r-1}$ is said to be \emph{nullified} at $\alpha$.

The theory of this subsection appears identical to the work in \citep{BDEMW13}.  The difference is in the application of the theory in Section \ref{sec:Algorithm}.  We suppose that the input is a list of QFFs, $\{\phi_i\}$, with each $A_i$ defined from the polynomials in each $\phi_i$.  In \citep{BDEMW13} there was an assumption (no longer made) that each of these formulae had a designated EC $f_i=0$ from which the subsets $E_i$ are defined.  Instead, we define $E_i$ to be a basis for $\{f_i\}$ if there is such a designated EC and define $E_i = A_i$ otherwise.  That is, we need to treat all the polynomials in QFFs with no EC with the importance usually reserved for ECs.

\subsection{Comparison with using a single implicit equational constraint}
\label{subsec:Implicit}

It is clear that in general the reduced projection $P_{\mathcal{E}}(\mathcal{A})$ will lead to fewer projection polynomials than using the full projection $P$.  However, a comparison with the existing theory of equational constraints requires a little more care.

First, we note that the TTICAD theory is applicable to a sequence of formulae while the theory of \cite{McCallum1999} is applicable only to a single formula.  Hence if the truth value of each QFF is needed then TTICAD is the only option; a truth-invariant CAD for a parent formula will not necessarily suffice.  Second we note that even if the sequence do form a parent formula then this must have an overall EC  to use \citep{McCallum1999} while the TTICAD theory is applicable even if this is not the case.

Let us consider the situation where both theories are applicable, i.e. we have a sequence of formulae (forming a parent formula) for which each has an EC and thus the parent formula  an implicit EC (their product).   In the context of Section \ref{subsec:TTICAD} this corresponds to using $\prod_i f_i$ as the EC.
The implicit EC approach would correspond to using the reduced projection $P_E(A)$ of \citep{McCallum1999}, with $E=\cup_i E_i$ and $A=\cup_i A_i$.  We make the simplifying assumption that $A$ is an irreducible basis.  In general $P_{\mathcal{E}}(\mathcal{A})$ will still contain fewer polynomials than $P_E(A)$ since $P_E(A)$ contains all resultants res$(f,g)$ where $f \in E_i, g \in A_j$ (and $g \notin E$), while $P_{\mathcal{E}}(\mathcal{A})$ contains only those with $i=j$ (and $g \notin E_i$).  Thus even in situations where the previous theory applies there is an advantage in using the new TTICAD theory.  These savings are highlighted by the worked examples in the next subsection and the complexity analysis later.

\subsection{Worked examples}
\label{subsec:WE2}

In Section \ref{sec:Algorithm} we define an algorithm for producing TTICADs. First we illustrate the savings with our worked examples from Section \ref{subsec:WE1}, which satisfy the simplifying assumptions from Section \ref{subsec:TTIProjOp}.

We start by considering $\Phi$ from equation (\ref{eqn:ExPhi}). In the notation above we have:
\begin{align*}
A_1 &:= \{f_1,g_1\}, \qquad E_1:=\{ f_1 \}; \\
A_2 &:= \{f_2,g_2\}, \qquad E_2:=\{ f_2 \}.
\end{align*}
We construct the reduced projection sets for each $\phi_i$,
\begin{align*}
P_{E_1}(A_1) &= \left\{ x^2-1, x^4 - x^2 + \tfrac{1}{16} \right\}, \\
P_{E_2}(A_2) &= \left\{ x^2 - 8x +15, x^4 -16x^3 + 95x^2-248x + \tfrac{3841}{16} \right\},
\end{align*}
and the cross-resultant set
\begin{equation*}
{\rm Res}^{\times} (\mathcal{E}) = \{{\rm res}_{y}(f_1,f_2)\} = \{ 68x^2 -272x + 285\}.
\end{equation*}
$P_{\mathcal{E}}(\mathcal{A})$ is then the union of these three sets.  In Figure \ref{fig:WE2} we plot the polynomials (solid curves) and identify the 12 real solutions of $P_{\mathcal{E}}(\mathcal{A})$ (solid vertical lines).  We can see the solutions align with the asymptotes of the $f_i$'s and the important intersections (those of $f_1$ with $g_1$ and $f_2$ with $g_2$).

\begin{figure}[p]
\caption{The polynomials from $\Phi$ in equation (\ref{eqn:ExPhi}) along with the roots of $P_{\mathcal{E}}(\mathcal{A})$ (solid lines), $P_E(A)$ (dashed lines) and $P(A)$ (dotted lines).}
\label{fig:WE2}
\begin{center}
\includegraphics[width=4.7in]{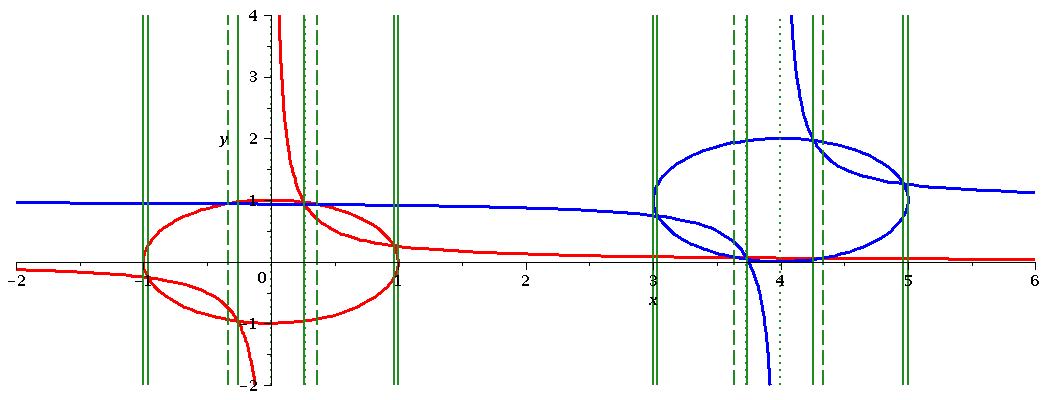}
\end{center}
\end{figure}

\begin{figure}[p]
\caption{Magnified region of Figure \ref{fig:WE2}.}
\label{fig:WE3}
\begin{center}
\includegraphics[width=2.5in]{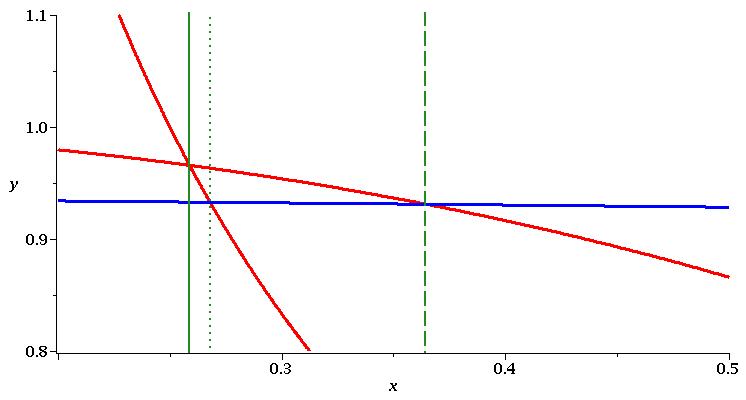}
\end{center}
\end{figure}

\begin{figure}[p]
\caption{The polynomials from $\Psi$ in equation (\ref{eqn:ExPsi}) along with the roots of $P_{\mathcal{E}}(\mathcal{A})$.}
\label{fig:WE4}
\begin{center}
\includegraphics[width=4.7in]{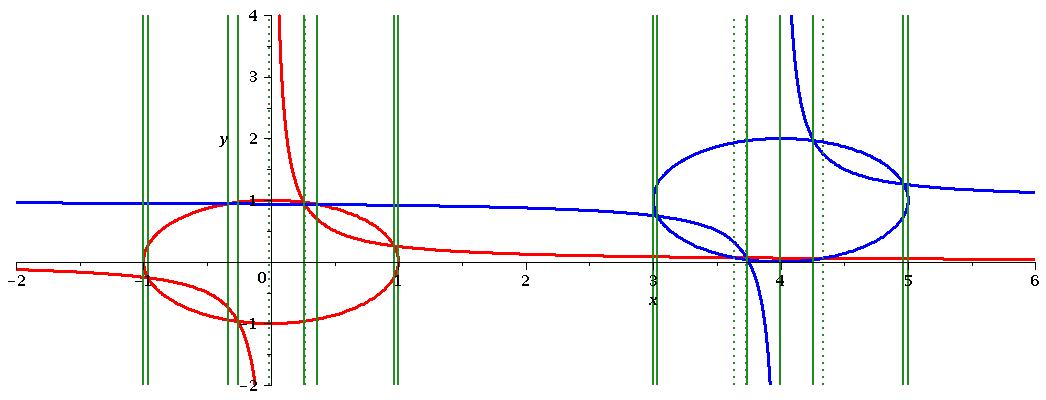}
\end{center}
\end{figure}

\begin{figure}[p]
\caption{Magnified region of Figure \ref{fig:WE4}.}
\label{fig:WE5}
\begin{center}
\includegraphics[width=2.5in]{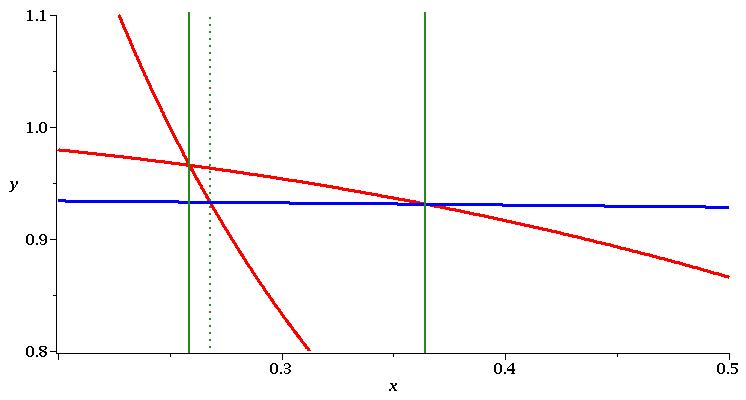}
\end{center}
\end{figure}

If we were to instead use a projection operator based on an implicit EC $f_1f_2=0$ then in the notation above we would construct $P_E(A)$ from $A=\{f_1,f_2,g_1,g_2\}$ and $E=\{f_1,f_2\}$.  This set provides an extra 4 solutions (the dashed vertical lines) which align with the intersections of $f_1$ with $g_2$ and $f_2$ with $g_1$.
Finally, if we were to consider $P(A)$ then we gain a further 4 solutions (the dotted vertical lines) which align with the intersections of $g_1$ and $g_2$ and the asymptotes of the $g_i$'s.
In Figure \ref{fig:WE3} we magnify a region
to show explicitly that the point of intersection between $f_1$ and $g_1$ is identified by $P_{\mathcal{E}}(\mathcal{A})$, while the intersections of $g_2$ with both $f_1$ and $g_1$ are ignored.

The 1-dimensional CAD produced using $P_{\mathcal{E}}(\mathcal{A})$ has 25 cells compared to 33 when using $P_E(A)$ and 41 when using $P(A)$.  However, it is important to note that this reduction is amplified after lifting (using Theorem \ref{thm:Main} and and Algorithm \ref{alg:TTICAD}).  The 2-dimensional TTICAD has 105 cells and the sign-invariant CAD has 317.  Using {\sc Qepcad} to build a CAD invariant with respect to the implicit EC gives us 249 cells.

Next we consider determining the truth of $\Psi$ from equation (\ref{eqn:ExPsi}).  This time
\begin{align*}
A_1 &:= \{f_1,g_1\}, \,\, E_1:=\{f_1\}, \\
A_2 &:= \{f_2,g_2\}, \,\, E_2:=\{f_2,g_2\},
\end{align*}
and so $P_{E_1}(A_1)$ is as above but $P_{E_2}(A_2)$ contains an extra polynomial $x-4$ (the coefficient of $y$ in $g_2$).
The cross-resultant set ${\rm RES}^{\times} (\mathcal{E})$ also contains an extra polynomial,
\[
{\rm res}_{y}(f_1,g_2) = x^4-8x^3+16x^2+\tfrac{1}{2}x-\tfrac{31}{16}.
\]
These two extra polynomials provide three extra real roots and hence the 1-dimensional CAD produced using $P_{\mathcal{E}}(\mathcal{A})$ this time has 31 cells.

In Figure \ref{fig:WE4} we again graph the four curves this time with solid vertical lines highlighting the real solutions of $P_{\mathcal{E}}(\mathcal{A})$.  By comparing with Figure \ref{fig:WE2} we see that more points in the CAD of $\R^1$ have been identified for the TTICAD of $\Psi$ than the TTICAD of $\Phi$ (15 instead of 12) but that there is still a saving over the sign-invariant CAD (which had 20, the five extra solutions indicated by dotted lines).
The lack of an EC in the second clause has meant that the asymptote of $g_2$ and its intersections with $f_1$ have been identified.
However, note that the intersections of $g_1$ with $f_2$ and $g_2$ and have not been.
Figure \ref{fig:WE5} magnifies a region of Figure \ref{fig:WE4}.  Compare with Figure \ref{fig:WE3} to see the dashed line has become solid, while the dotted line remains unidentified by the TTICAD.

Note that we are unable to use \citep{McCallum1999} to study $\Psi$ as there is no polynomial equation logically implied (either explicitly or implicitly) by this formula.  Hence there are no dashed lines and the choice is between the sign-invariant CAD with 317 cells or the TTICAD, which for this example has 183 cells.

\section{Algorithm}
\label{sec:Algorithm}

\subsection{Description and Proof}
\label{subsec:Algorithm}

We describe carefully Algorithm \ref{alg:TTICAD}.  This will create a TTICAD of $\R{}^n$
for a list of QFFs $\{ \phi_i \}_{i=1}^t$ in variables ${\bf x} = x_1 \prec x_2 \prec \cdots \prec x_n$, where each $\phi_i$ has at most one designated EC $f_i = 0$ of positive degree (there may be other non-designated ECs).

It uses a subalgorithm \texttt{CADW}, which was validated by \cite{McCallum1998}.  The input of {\tt CADW} is: $r$, a positive integer and $A$, a set of $r$-variate integral polynomials. The output is a boolean $w$ which if true is accompanied by an order-invariant CAD for $A$ (represented as a list of indices $I$ and sample points $S$).

Let $A_i$ be the set of all polynomials occurring in $\phi_i$.  If $\phi_i$ has a designated EC then put $E_i = \{f_i\}$ and if not put $E_i=A_i$.
Let $\mathcal{A}$ and $\mathcal{E}$ be the lists of the $A_i$ and $E_i$ respectively.
Our algorithm effectively defines the reduced projection of $\mathcal{A}$ with respect to $\mathcal{E}$ in terms of the special case of this definition from the previous section.
The definition amounts to
\begin{equation}
\label{eqn:TTIProj-G}
P_{\mathcal{E}}(\mathcal{A}) := C \cup P_{\mathcal{F}}(\mathcal{B}).
\end{equation}
Here $C$ is the set of contents of all the elements of all $A_i$; $\mathcal{B}$ the list $\{B_i\}_{i=1}^t$ such that $B_i$ is the finest\footnote{A decomposition into irreducibles. This avoids various technical problems.} squarefree basis for the set ${\rm prim}(A_i)$ of primitive parts of elements of $A_i$ which have positive degree; and $\mathcal{F}$ is the list $\{F_i\}_{i=1}^t$, such that $F_i$ is the finest squarefree basis for ${\rm prim}(E_i)$.
(The reader may notice that this notation and the definition of $P_\mathcal{E}(\mathcal{A})$ here is analogous to the work in Section 5 of \citep{McCallum1999}.)

\begin{algorithm}[ht]
\label{alg:TTICAD}
\SetKwInOut{Input}{Input}\SetKwInOut{Output}{Output}
\Input{A list of quantifier-free formulae $\{ \phi_i \}_{i=1}^t$ in variables $x_1,\ldots,x_n$. Each $\phi_i$ has at most one designated EC $f_i = 0$.
}
\Output{Either
$\bullet$ $\mathcal{D}:$ A CAD of $\R{}^n$ (described by lists $I$ and $S$ of cell indices and sample points) which is truth table invariant for the list of input formulae; or \quad
$\bullet$~{\bf FAIL}: If $\mathcal{A}$ is not well-oriented with respect to $\mathcal{E}$ (Def \ref{def:WO}).
}
\BlankLine
\For{$i = 1 \dots t$}{
 If there is no designated EC then set $E_i:=A_i$ and otherwise set $E_i:=\{f_i\}$\;
 Compute the finest squarefree basis $F_i$ for ${\rm prim}(E_i)$\;
}
Set $F \leftarrow \cup_{i=1}^t F_i$\;
\eIf{$n=1$}{
 Isolate the real roots of the polynomials in $F$ and thus form cell indices and sample points for a CAD of $\R$
 \label{step:base1}\;
 \Return $I$ and $S$ for $\mathcal{D}$
 \label{step:base2}\;
}
{
 \For{$i = 1 \dots t$}{
   Extract the set $A_i$ of polynomials in $\phi_i$\;
   Compute the set $C_i$ of contents of the elements of $A_i$\;
   Compute the set $B_i$, the finest squarefree basis for ${\rm prim}(A_i)$\;
 }
 Set $C := \cup_{i=1}^t C_i$,
 $\mathcal{B} := (B_i)_{i=1}^t$ and
 $\mathcal{F} := (F_i)_{i=1}^t$\;
 Construct the projection set $\mathfrak{P} := C \cup P_{\mathcal{F}}(\mathcal{B})$
 \label{step:mfP}\;
 Attempt to construct a lower-dimensional CAD:
 $w',I',S' := {\tt CADW}(n-1,\mathfrak{P})$
 \label{step:cadw}\;
 \If{$w' = false$}{
   \Return {\bf FAIL} (since $\mathfrak{P}$ is not well oriented)
   \label{step:notWO1}\;
 }
 $I \leftarrow\emptyset$; $S \leftarrow\emptyset$
 \label{step:lifting1}\;
 \For{each cell $c \in \mathcal{D}'$}{
   $L_c \leftarrow \{\}$\;
   \For{$i = 1,\ldots t$}{
     \eIf{any $f \in E_i$ is nullified on $c$ \label{step:CleverCheck}}{
       \eIf{$\dim(c)>0$}{
         \Return {\bf FAIL} (since $\{ \phi_i \}_{i=1}^t$ is not well oriented)
         \label{step:notWO2}\;
       }{
         $L_c \leftarrow L_c \cup B_i$
         \label{step:addthebi}\;
       }
     }{
       $L_c \leftarrow L_c \cup F_i$\;
     }
   }
   Generate a stack over $c$ using $L_c$: construct cell indices and sample points for the stack over $c$ of the  polynomials in $L_c$, adding them to $I$ and $S$
   \label{step:lifting2}\;
 }
 \Return $I$ and $S$ for $\mathcal{D}$\;
}
\caption{TTICAD Algorithm}
\end{algorithm}

We shall prove that, provided the input satisfies the condition of well-orientedness given in Definition \ref{def:WO}, the output of  Algorithm \ref{alg:TTICAD} is indeed a TTICAD for $\{\phi_i\}$.
We first recall the more general notion of well-orientedness from \citep{McCallum1998}.  The boolean output of \texttt{CADW} is false if the input set was not well-oriented in this sense.
\begin{defn}\label{def:WO-original}
A set $A$ of $n$-variate polynomials is said to be \emph{well oriented} if whenever $n > 1$,
every $f \in {\rm prim}(A)$ is nullified by at most a finite number of points
in $\R^{n-1}$, and (recursively) $P(A)$ is well-oriented.
\end{defn}

This condition is required for \texttt{CADW} since the validity of this algorithm relies on Theorem \ref{thm:McC1} which holds only when polynomials do not vanish identically.  The conditions allows for a finite number of these nullifications since this indicates a problem on a zero cell, that is a single point.  In such cases it is possible to replace the nullified polynomial by a so called \emph{delineating polynomial} which is not nullified and can be used in place to ensure the delineability of the other.  The use of these is part of the verified algorithm \texttt{CADW} \citep{McCallum1998} and they are studied in detail by \cite{Brown2005a}.

We now define our new notion of well-orientedness for the lists of sets $\mathcal{A}$ and $\mathcal{E}$.
\begin{defn}\label{def:WO}
We say that $\mathcal{A}$ is {\em well oriented with respect to} $\mathcal{E}$ if, whenever $n > 1$,
every polynomial $f \in E$ is nullified by at most a finite number of points in $\R^{n-1}$, and $P_{\mathcal{F}}(\mathcal{B})$ is well-oriented in the sense of Definition \ref{def:WO-original}.
\end{defn}

It is clear than Algorithm \ref{alg:TTICAD} terminates.  We now prove that it is correct using the theory developed in Section \ref{sec:TTIProj}.

\begin{thm}
The output of Algorithm \ref{alg:TTICAD} is as specified.
\end{thm}

\begin{proof}
We must show that when the input is well-oriented the output is a TTICAD, (each $\phi_i$ has constant truth value in each cell of $\mathcal{D}$), and \textbf{FAIL} otherwise.

If the input was univariate then it is trivially well-oriented.  The algorithm will construct a CAD $\mathcal{D}$ of $\R^1$ using the roots of the irreducible factors of the polynomials in $E$ (steps \ref{step:base1} to \ref{step:base2}).
At each 0-cell all the polynomials in each $\phi_i$ trivially have constant signs, and hence every $\phi_i$ has constant truth value.  In each 1-cell no EC can change sign and so every $\phi_i$ has constant truth value $false$, unless there are no ECs in any clause.  In this case the algorithm would have constructed a CAD using all the polynomials and hence on each 1-cell no polynomial changes sign and so each clause has constant truth value.

From now on suppose $n > 1$.  If $\mathfrak{P} = C \cup P_{\mathcal{F}}(\mathcal{B})$ is not well-oriented in the sense of Definition \ref{def:WO-original} then \texttt{CADW} returns $w'$ as false.  In this case the input is not well oriented in the sense of Definition \ref{def:WO} and Algorithm \ref{alg:TTICAD} correctly returns \textbf{FAIL} in step \ref{step:notWO1}.
Otherwise, 
we have $w'= true$ with $I'$ and $S'$ specifying a CAD, $\mathcal{D}'$, which is order-invariant with respect to $\mathfrak{P}$ (by the correctness of \texttt{CADW}, as proved in \citep{McCallum1998}).
Let $c$, a submanifold of $\R{}^{n-1}$, be a cell of $\mathcal{D}'$ and let $\alpha$ be its sample point.

We suppose first that the dimension of $c$ is positive.
If any polynomial $f \in E$ vanishes identically on $c$ then the input is not well oriented in the sense of Definition \ref{def:WO} and the algorithm correctly returns \textbf{FAIL} at step \ref{step:notWO2}.
Otherwise, we know that the input list was certainly well-oriented.  Since no polynomial $f \in E$ vanishes then no element of the basis $F$ vanishes identically on $c$ either.  Hence, by Theorem \ref{thm:Main}, applied with $\mathcal{A} = \mathcal{B}$ and
$\mathcal{E} = \mathcal{F}$, each element of $F$ is delineable on $c$, and
the sections over $c$ of the elements of $F$ are pairwise disjoint.
Thus the sections and sectors over $c$ of the elements of $F$ comprise
a stack $\Sigma$ over $c$.  Furthermore, the last conclusion of Theorem \ref{thm:Main} assures us that, for each $i$, every element of $B_i \setminus F_i$ is sign-invariant in each section over $c$ of every element of $F_i$.  Let $1 \le i \le t$.  We shall show that each $\phi_i$ has constant truth value in both the sections and sectors of $\Sigma$.

If $\phi_i$ has a designated EC then let $f_i$ denote the constraint polynomial; otherwise let $f_i$ denote an arbitrary element of $A_i$.

Consider first a section $\sigma$ of $\Sigma$. Now $f_i$ is a product of its content ${\rm cont}(f_i)$ and some elements of the basis $F_i$.  But ${\rm cont}(f_i)$, an element of $\mathfrak{P}$, is sign-invariant (indeed order-invariant) in the whole cylinder $c \times \R$ and hence, in particular, in $\sigma$. Moreover all of the elements of $F_i$ are sign-invariant in $\sigma$, as was noted previously. Therefore $f_i$ is sign-invariant in $\sigma$.  If $\phi_i$ has no constraint (and so $f_i$ denotes an arbitrary element of $A_i$) then this implies that $\phi_i$ has constant truth value in $\sigma$. So consider from now on the case in which $f_i = 0$ is the designated constraint polynomial of $\phi_i$.

If $f_i$ is positive or negative in $\sigma$ then $\phi_i$ has constant truth value $false$ in $\sigma$.  So suppose that $f_i = 0$ throughout $\sigma$. It follows that $\sigma$ must be a section of some element of the basis $F_i$.
Let $g \in A_i \setminus E_i$ be a non-constraint polynomial in $A_i$.
Now, by the definition of $B_i$, we see $g$ can be written as
\[
g = {\rm cont}(g) h_1^{p_1} \cdots h_k^{p_k}
\]
where $h_j \in B_i, p_j \in \mathbb{N}$.
But ${\rm cont}(g)$, in $\mathfrak{P}$, is sign-invariant (indeed order-invariant) in the whole cylinder $c \times \R$, and hence in particular in $\sigma$.  Moreover each $h_j$ is sign-invariant in $\sigma$, as was noted previously.  Hence $g$ is sign-invariant in $\sigma$.  (Note that in the case where $g$ does not have main variable $x_n$ then $g = {\rm cont}(g)$ and the conclusion still holds).  Since $g$ was an arbitrary element of $A_i \setminus E_i$, it follows that all polynomials in $A_i$ are sign-invariant in $\sigma$, hence that $\phi_i$ has constant truth value in $\sigma$.

Next consider a sector $\sigma$ of the stack $\Sigma$, and notice that at least one such sector exists.
As observed above, ${\rm cont}(f_i)$ is sign-invariant in $c$, and $f_i$ does not vanish identically on $c$. Hence ${\rm cont}(f_i)$ is non-zero throughout $c$. Moreover each element of the basis $F_i$ is delineable on $c$. Hence $f_i$ is nullified by no point of $c$.  It follows from this that the algorithm does not return \textbf{FAIL} during the lifting phase.  It follows also that $f_i \neq 0$ throughout $\sigma$. Hence $\phi_i$ has constant truth value $false$ in $\sigma$.

It remains to consider the case in which the dimension of $c$ is 0.  In this case the roots of the polynomials in the lifting set $L_c$ constructed by the algorithm determine a stack $\Sigma$ over $c$.  Each $\phi_i$ trivially has constant truth value in each section (0-cell) of this stack, and the same can routinely be shown for each sector (1-cell) of this stack.
\end{proof}

\subsection{TTICAD via the ResCAD Set}
\label{subsec:ResCAD}

When no $f \in E$ is nullified there is an alternative implementation of TTICAD which would be simple to introduce into existing CAD implementations.
Define 
\begin{equation*}
  \mathcal{R}(\{ \phi_i\}) = E \cup {\textstyle \bigcup_{i=1}^t} \left\{ {\rm res}_{x_n}(f,g) \mid f\in E_i, g \in A_i, g \notin E_i \right\}.
\end{equation*}
to be the {\em ResCAD Set} of $\{\phi_i\}$.

\begin{thm}
\label{thm:ResCAD}
   Let $\mathcal{A} = ( A_i)_{i=1}^t$ be a list of irreducible bases $A_i$
and let $\mathcal{E} = ( E_i )_{i=1}^t$ be a list of non-empty subsets
$E_i \subseteq A_i$.   Then we have
\begin{equation*}
{P}(\mathcal{R}(\{ \phi_i\})) = {P}_{\mathcal{E}}(\mathcal{A}).
\end{equation*}
\end{thm}
\noindent The proof is straightforward and so omitted here.

\begin{cor}
\label{cor:ResCAD}
If no $f \in E$ is nullified by a point in $\R{}^{n-1}$ then inputting $\mathcal{R}(\{ \phi_i\})$ into any algorithm which produces a sign-invariant CAD using McCallum's projection operator $P$ will result in the TTICAD for $\{ \phi_i\}$ produced by Algorithm \ref{alg:TTICAD}.
\end{cor}

Corollary \ref{cor:ResCAD} gives a simple way to compute TTICADs using existing CAD implementations based on McCallum's approach, such as {\sc Qepcad}.

\section{Utilising projection theory for improvements to lifting}
\label{sec:ImprovedLifting}

Consider the case when the input to Algorithm \ref{alg:TTICAD} is a single QFF $\{\phi\}$ with a declared EC.  In this case the reduced projection operator $P_{\mathcal{E}}(\mathcal{A})$ produces the same polynomials as the operator $P_{E}(A)$ and so one may expect the TTICAD produced to be the same as the CAD 
produced by an implementation of \citep{McCallum1999} such as {\sc Qepcad}.
In practice this is not the case because Algorithm \ref{alg:TTICAD} makes use of the reduced projection theory in the lifting phase as well as the projection phase.

\cite{McCallum1999} discussed how the theory of a reduced projection operator would improve the projection phase of CAD, by creating fewer projection polynomials.  The only modification to the lifting phase of Collins' CAD algorithm described was the need to check the well-orientedness condition of Definition \ref{def:WO-original}.

In this section we note two subtleties in the lifting phase of Algorithm \ref{alg:TTICAD} which result in efficiencies that could be replicated for use with the original theory.  In fact, the \texttt{ProjectionCAD} package \citep{EWBD14} discussed in Section \ref{subsec:PCAD} has commands for building CADs invariant with respect to a single EC which does this.

\subsection{A finer check for well-orientedness}
\label{subsec:IL_finerWO}

Theorem 2.3 of \citep{McCallum1999} verified the use of $P_E(A)$.  The proof uses Theorem \ref{thm:McC1} to conclude sign-invariance for the polynomial defining the EC, and Theorem \ref{thm:McC2} to conclude sign-invariance for the other polynomials only when the EC was satisfied.

To apply Theorem \ref{thm:McC1} here we need the EC polynomial and the projection polynomials obtained by repeatedly applying $P$ to have a finite number of nullification points.  Meanwhile, the application of Theorem \ref{thm:McC2} requires that the resultants of the EC polynomial with the others polynomials have no nullification points.  Both these requirements are guaranteed by the input satisfying Definition \ref{def:WO-original}, the condition used in \citep{McCallum1999}.  However, this also requires that other projection polynomials, including the non-ECs in the input, to have no nullification points.  

In Algorithm \ref{alg:TTICAD}, step \ref{step:CleverCheck} only checks for nullification of the polynomials in $E_i$ (in this context meaning only the EC).
Hence this algorithm is checking the necessary conditions but not whether the non-ECs (in the main variable) are nullified.

\begin{exmp}
Assume the variable ordering $x \prec y \prec z \prec w$ and consider the polynomials
\[
f = x+y+z+w, \qquad g = zy - x^2w
\]
forming the formula $f=0 \wedge g<0$.
We could analyse this using a sign-invariant CAD
with 557 cells but it is more efficient to make use of the EC.  Our implementation of Algorithm \ref{alg:TTICAD} produces a CAD with 165 cells, while declaring the EC in QEPCAD results in 221 cells (the higher number is explained in subsection \ref{subsec:IL_smallerL}).  \textsc{Qepcad} also prints:
\begin{verbatim}
      Error! Delineating polynomial should be added over cell(2,2)!
\end{verbatim}
indicating the output may not be valid.  The error message was triggered by the nullification of $g$ when $x=y=0$ which does not actually invalidate the theory. \textsc{Qepcad} is checking for nullification of all projection polynomials leading to unnecessary errors.
\end{exmp}
In fact, we can take this idea further in the case where $E_i=A_i$ for some $i$: in such a case we do not need to check any elements of (that particular) $E_i$ for nullification (since we are using the theory of \cite{McCallum1998} and it is the final lift meaning only sign- (rather than order-) invariance is required.  

\subsection{Smaller lifting sets}
\label{subsec:IL_smallerL}

Traditionally in CAD algorithms the projection phase identifies a set of projection polynomials, which are then used in the lifting phase to create the stacks.  However when making use of ECs we can actually be more efficient by discarding some of the projection polynomials before lifting.
The non-ECs (in the main variable) are part of the set of projection polynomials, required in order to produce subsequent projection polynomials (when we take their resultant with the EC).  However, these polynomials are not (usually) required for the lifting since Theorem \ref{thm:McC2} can (usually) be used to conclude them sign-invariant in those sections produced when lifting with the EC.

Note that in Algorithm \ref{alg:TTICAD} the projection polynomials are formed from the input polynomials (in the main variable) and the set of polynomials $\mathfrak{P}$ constructed in step \ref{step:mfP} which are not in the main variable.  The  lower dimensional CAD $D$ constructed in step \ref{step:cadw} is guaranteed to be sign-invariant
for $\mathfrak{P}$.  In particular, $\mathfrak{P}$ contains the resultants of the EC with the other constraints and thus $D$ is already decomposing the domain into cells such that the presence of an intersection of $f$ and $g$ is invariant in each cell.  Hence for the final lift we need to build stacks with respect to $f$.

The following examples demonstrate these efficiencies.

\begin{exmp}
Consider from Section \ref{subsec:WE1} the circle $f_1$, hyperbola $g_1$ and sub-formula $\phi_1 := f_1=0 \land g_1<0$.  Building a sign-invariant CAD for these polynomials
uses 83 cells with the induced CAD of $\R$ identifying 7 points.
Declaring the EC in QEPCAD results in a CAD with 69 cells while using our implementation of Algorithm \ref{alg:TTICAD} produces a CAD with 53 cells.
Both implementations give the same induced CAD of $\R$ identifying 6 points but \textsc{Qepcad} uses more cells for the CAD of $\R^2$.  

In particular, \texttt{ProjectionCAD} has a cell where $x<-2$ and $y$ is free while \textsc{Qepcad} uses three cells, splitting where $g_1$ changes sign.  The splitting is not necessary for a CAD invariant with respect to the EC since $f_1$ is non-zero (and $\phi_1$ hence false) for all $x<-2$.
\end{exmp}

\begin{exmp}
Now consider all four polynomials from Section \ref{subsec:WE1} and the formula $\Phi$ from equation (\ref{eqn:ExPhi}).  In Section \ref{subsec:WE2} we reported that a TTICAD could be built with 105 cells compared to a CAD with 249 cells built invariant with respect to the implicit EC $f_1f_2=0$ using {\sc Qepcad}.  The improved projection resulted in the induced CAD of $\R$ identifying 12 points rather than 16.

We now observe that some of the cell savings was actually down to using smaller sets of lifting polynomials.  We may simulate the projection with respect to the implicit EC via Algorithm \ref{alg:TTICAD} by inputting a set consisting of the single formula
\[
\Phi' = f_1f_2=0 \land \Phi
\]
(note that logically $\Phi = \Phi'$).  The implementation in \texttt{ProjectionCAD} would then produce a CAD with 145 cells.  So we may conclude that improved lifting allowed for a saving of 
104 cells and improved projection a further saving of
40 cells.

\end{exmp}

In this example 72\% of the cell saving came from improved lifting and 28\% from improved projection, but we should not conclude that the former is more important.  The improvement is to the final lift (from a CAD of $\R^{n-1}$ to one of $\R^n$) and the first projection (from polynomials in $n$ variables to those with $n-1$).  Hence the savings from improved projection get magnified throughout the rest of the algorithm, and so as the number of variables in a problem increases so will the importance of this.

\begin{exmp}
\label{ex:3dPhi}
We consider a simple 3d generalisation of the previous example.  Let
\begin{align*}
\Phi^{3d} &= \big( x^2+y^2+z^2-1=0 \land xyz - \tfrac{1}{4}<0 \big) \\
&\qquad \lor \big( (x-4)^2+(y-1)^2+(z-2)^2-1=0 \land (x-4)(y-1)(z-2) - \tfrac{1}{4}<0 \big)
\end{align*}
and assume variable ordering $x \prec y \prec z$. Using Algorithm \ref{alg:TTICAD} on the two QFFs joined by disjunction gives a CAD with 109 cells while declaring the implicit EC in {\sc Qepcad} gives 739 cells.  Using Algorithm \ref{alg:TTICAD} on the single formula conjuncted with the implicit EC gave a CAD with 353 cells.  So in this case the improved lifting saves
386 cells and the improved projection a further
244 cells.
\end{exmp}

Moving from 2 to 3 variables has increased the proportion of the saving from improved projection from 28\% to 39\%.  The complexity analysis in the next section will further demonstrate the importance of improved projection, especially for the problem classes where no implicit EC exists (see also the experiments in Section \ref{subsec:IncreasedBenefit}).

\section{Complexity analyses of new contributions}
\label{sec:CA}

In this Section we closely follow the approach of our new analysis for the existing theory given in Section \ref{subsec:CA1}.  We will first study the special case of TTICAD when every QFF has an EC, before moving to the general case.  This is because such formulae may be studied using \cite{McCallum1999} and so our comparison must be with this as well as \cite{McCallum1998} in order to fully clarify the advantages of our new projection operator.

\subsection{When every QFF has an equational constraint}
\label{subsec:CA-TTI1}

We consider a sequence of $t$ QFFs which together contain $m$ constraints and are thus defined by at most $m$ polynomials.  We suppose further that each QFF has at least one EC, and that the maximum degree of any polynomial in any variable is $d$.
Let $\mathcal{A}$ be the sequence of sets of polynomials $A_i$ defining each formula, $\mathcal{E}$ the sequence of subsets $E_i \subset A_i$ defining the ECs, and denote the irreducible bases of these by $B_i$ and $F_i$.
 
\begin{lem}
\label{L:TTICAD1}
Under the assumptions above, $P_{\mathcal{E}}(\mathcal{A})$ has the $(M,2d^2)$-property with  
\begin{equation}
\label{eq:TTI-M}
M = \left\lfloor \tfrac{1}{2}(3m+1) \right\rfloor + \tfrac{1}{2}(t-1)t.
\end{equation}
\end{lem}

\begin{proof}    
From equations (\ref{eqn:TTIProj-G}) and (\ref{eqn:TTIProj-S}) we have
\begin{equation}
P_{\mathcal{E}}(\mathcal{A}) = \cont(\mathcal{A}) \cup \textstyle{\bigcup_{i=1}^t} P_{F_i}(B_i) \cup {\rm Res}^{\times} (\mathcal{F}).
\label{eq:TTIProj}
\end{equation}

\begin{enumerate}
\item Consider first the cross resultant set.  Let $T_1$ be the set of elements of $B_i$ which divide some element of $F_1$, and $T_i, i=2, \dots, t$ be those elements of $B_i$ which divide some element of $F_i$ and do not already occur in some $T_j: j<i$.   Then using the same argument as in the proof of Lemma \ref{L:SINew} step 2 we see that the cross-resultant set can be partitioned into $\tfrac{1}{2}(t-1)t$ sets of combined degrees at most $2d^2$.
\item We now consider the $P_{E_i}(A_i)$ since
\begin{equation}
\label{eq:ttiProof}
\cont(\mathcal{A}) \cup \textstyle{\bigcup_{i=1}^t} P_{F_i}(B_i) = \textstyle{\bigcup_{i=1}^t} P_{E_i}(A_i).
\end{equation}
\begin{enumerate}
\item Let $m_i$ be the polynomials defining $A_i$.  We follow Lemma \ref{L:EC} to say that for each $i$: the contents, leading coefficients and discriminants for $E_i$ form a set $R_{i,1}$ with combined degree $2d^2$; the other coefficients for $E_i$ form a set $R_{i,2}$ with combined degree $d^2$; the remaining contents of each $A_i$ form a set $R_{i,3} = \cont(A_i) \setminus \cont(R_{i,1})$ with the $(m_i-1,d^2)$-property; the final set of resultants in (\ref{eq:ECProj}) for each $i$ form a set $R_{i,4}$ with the $(m_i-1, 2d^2)$-property.
\item $R_1 = \textstyle{\bigcup_{i=1}^t} R_{i,1}$ has the $(t, 2d^2)$-property while 
$R_4 = \textstyle{\bigcup_{i=1}^t} R_{i,4}$ may be partitioned into 
$\textstyle{\sum_{i=1}^t} m_i -1 = m - t$
sets of combined degree $2d^2$.
\item The union $R_{23} = \textstyle{\bigcup_{i=1}^t} R_{i,2} \cup R_{i,3}$ may be partitioned into 
\[
\textstyle{\sum_{i=1}^t} m_i - 1 + 1 = m
\]
sets of combined degree $d^2$, and so has the $\big( \lfloor \tfrac{1}{2}(m+1) \rfloor, 2d^2\big)$-property.
\end{enumerate}
Hence (\ref{eq:ttiProof}), which equals $R_1 \cup R_{23} \cup R_4$, has the 
$\left( \left\lfloor \tfrac{1}{2}(3m+1) \right\rfloor, 2d^2 \right)$ property. 
\end{enumerate}
So together we see that (\ref{eq:TTIProj}) has the $(M,2d^2)$-property with $M$ as given in (\ref{eq:TTI-M}).
\end{proof}

To analyse Algorithm \ref{alg:TTICAD} we will apply Lemma \ref{L:TTICAD1} once and then Corollary \ref{cor:SI2} repeatedly.  The growth in factors is given by Table \ref{tab:GeneralProjection}, with $M$ this time representing (\ref{eq:TTI-M}).  Thus the dominant term in the bound is calculated from (\ref{bound:All}) (omitting the floor in $M$) as
\begin{align}
&\qquad 2^{2^{n}-1}d^{2^{n}-1}( \tfrac{1}{2}(3m+1) + \tfrac{1}{2}(t-1)t )^{2^{n-1}-1}m \nonumber \\
&= 2^{2^{n-1}}d^{2^{n}-1}(3m+t^2-t+1)^{2^{n-1}-1}m. \label{bound:TTIbasic}
\end{align}
Actually, this bound can be lowered by noting that for the final lift we use only the $t$ ECs rather than all $m$ of the input polynomials, reducing the bound to 
\begin{equation}
\label{bound:TTI1}
2^{2^{n-1}}d^{2^{n}-1}(3m+t^2-t+1)^{2^{n-1}-1}t.
\end{equation}

\begin{rem}
Observe that if $t=1$ then the value of $M$ for TTICAD in (\ref{eq:TTI-M}) becomes (\ref{eq:EC-M}), the value for a CAD invariant with respect to an EC.  Similarly, if $t=m$ then (\ref{eq:TTI-M}) becomes (\ref{eq:SI-M}), the value for sign-invariant CAD.  Actually, in these two situations the TTICAD projection operator reverts to the previous ones.  These are the extremal values of $t$ and provide the best and worse cases respectively.
\end{rem}

We can conclude from the remark that TTICAD is superior to sign-invariant CAD (strictly so unless $t=m$).  Comparing the bounds (\ref{bound:TTI1}) and (\ref{bound:SINew}) we see the effect is a reduction in the double exponent of the factor dependent on $m$ for $t \ll m$, which gradually reduces as $t$ gets closer to $m$.

It would be incorrect to conclude from the remark that the theory of \cite{McCallum1999} is superior to TTICAD.  In the case $t=1$ the algorithms and their analysis are equal up to the final lifting stage.  As discussed in Section \ref{sec:ImprovedLifting} this can be applied to the case $t=1$ also, with the effect of reducing the bound (\ref{bound:EC}) by a factor of $m$ to
\begin{equation}
\label{bound:ECImproved}
2^{2^{n-1}}d^{2^{n}-1}(3m+1)^{2^{n-1}-1}.
\end{equation}
If $t>1$ then \cite{McCallum1999} cannot be applied directly since it requires a single formula with an EC.  However, it can be applied indirectly by considering the parent formula formed by the disjunction of the individual QFFs which has the product of the individual ECs as an implicit EC.  A CAD for this parent formula produced using \cite{McCallum1999} would also be a TTICAD for the sequence of QFFs.  Thus we provide a complexity analysis for this case.

\subsubsection{With a parent formula and implicit EC-CAD}

By working with the extra implicit EC we are starting with one extra polynomial, whose degree is $td$.  However, we know the factorisation into $t$ polynomials so suppose we start from here (indeed, this is what our implementation does).  

\begin{lem}
\label{L:ECImplicit}
Consider a set $A$ of $m$ polynomials in $n$ variables with maximum degree $d$, and a subset $E=\{f_1, \dots, f_t\} \subseteq A$.  Then $P_E(A)$, has the $(M,2d^2)$-property with 
\begin{equation}
\label{eq:ECImplicit-M}
M = \tfrac{1}{2}(2m-t+1)t + \left\lfloor \tfrac{1}{2}(m+1) \right\rfloor
\end{equation}
\end{lem}
\begin{proof} 
Partition $E$ into subsets $S_i=\{f_i\}$ for $i=1, \dots, t$.  Then $P_E(A)$ from (\ref{eq:ECProj}) is
\begin{align}
&{\textstyle \cont(A \setminus E) + \bigcup_{i=1}^t} P(S_i) 
+ \{ {\rm res}_{x_n}(f,g) \mid f \in F, g \in F, g \neq f \} \nonumber \\
&\qquad  + \{ {\rm res}_{x_n}(f,g) \mid f \in F, g \in B \setminus F \}. \label{eq:ECImpProof}
\end{align}
\begin{enumerate}
\item We start by considering the first two terms in (\ref{eq:ECImpProof}).
\begin{enumerate}
\item For each $P(S_i)$: the contents, leading coefficients and discriminants form a set $R_{i,1}$ with combined degree $2d^2$, and the other coefficients a set $R_{i,2}$ with combined degree $d^2$.  
\item The remaining contents $R_3 = \cont(A) \setminus \cont(E)$ has the $(m-t,d^2)$-property.  
\item Together, the set $R_1 = {\textstyle \bigcup_{i=1}^t} R_{1,i}$ has the $(t, 2d^2)$-property.
\item Together, $R_{23} = R_3 \cup {\textstyle \bigcup_{i=1}^t} R_{2,i}$ has the $(m,d^2)$-property.  It can be further partitioned into $\lfloor \tfrac{1}{2}(m+1) \rfloor$ sets of combined degree $2d^2$.
\end{enumerate}
The first two terms of (\ref{eq:ECImpProof}) may be partitioned into $R_1 \cup R_{23}$ and thus further into $t + \lfloor \tfrac{1}{2}(m+1) \rfloor$ sets of combined degree $2d^2$.
\item The first set of resultants in (\ref{eq:ECImpProof}) has size $\tfrac{1}{2}(t-1)t$ and maximum degree $2d^2$.  
\item The second set of resultants in (\ref{eq:ECImpProof}) may be decomposed as 
\[
{\textstyle \bigcup_{i=1}^t} \{ {\rm res}_{x_n}(f,g) \mid f \in S_i, g \in B \setminus F \}.
\]
Since $|S_i|=1$ and $|B \setminus F|$ has the $(m-t,d)$-property, each of these subsets has $(m-t, 2d^2)$-property (following Lemma \ref{L:SINew} step 2).  Thus together the set of them has the $(t(m-t), 2d^2)$-property.
\end{enumerate}
Hence $P_E(A)$ as given in (\ref{eq:ECImpProof} may be partitioned into
\[
t + \lfloor \tfrac{1}{2}(m+1) \rfloor + \tfrac{1}{2}(t-1)t + t(m-t) 
= \tfrac{1}{2}(2m-t+1)t + \left\lfloor \tfrac{1}{2}(m+1) \right\rfloor
\]
sets of combined degree $2d^2$.
\end{proof}

Thus the growth of projection polynomials in this case is given by Table \ref{tab:GeneralProjection} with $M$ from (\ref{eq:ECImplicit-M}).  The dominant term in the cell count bound is calculated from (\ref{bound:All}) as
\begin{align*}
&\qquad 2^{2^{n}-1}d^{2^{n}-1}( \tfrac{1}{2}(t(2m-t+1)+m+1) )^{2^{n-1}-1}m \nonumber \\
&= 2^{2^{n-1}}d^{2^{n}-1}(t(2m-t+1)+m+1)^{2^{n-1}-1}m.
\end{align*}
If we follow Section \ref{sec:ImprovedLifting} to simplify the final lift this reduces to 
\begin{equation}
\label{bound:ImplicitEC2}
2^{2^{n-1}}d^{2^{n}-1}(t(2m-t+1)+m+1)^{2^{n-1}-1}t.
\end{equation}

\subsubsection{Comparison}

Observe that if $t=1$ then the value of $M$ in (\ref{eq:ECImplicit-M}) becomes (\ref{eq:EC-M}), while if $t=m$ it becomes (\ref{eq:SI-M}), just like TTICAD.  However, since the difference between (\ref{eq:ECImplicit-M}) and (\ref{eq:TTI-M}) is 
\[
mt - t^2 - m + t = (t-1) (m-t).
\]
we see that for all other possible values of $t$ the TTICAD projection operator has a superior $(m,d)$-property.  This means fewer polynomials and a lower cell count, as noted earlier in Section \ref{subsec:Implicit}.  Comparing the bounds (\ref{bound:TTI1}) and (\ref{bound:ImplicitEC2}) we see the effect is a reduction in the base of the doubly exponential factor dependent on $m$.

\subsection{A general sequence of QFFs}
\label{subsec:CA-TTI2}

We again consider $t$ QFFs formed by at a set of at most $m$ polynomials with maximum degree $d$, however, we no longer suppose that each QFF has an EC.  Instead we denote by $\mathfrak{e}$ the number of QFFs with one; by $A_{\mathfrak{e}}$ the set of polynomials required to define those $\mathfrak{e}$ QFFs; and by $m_\mathfrak{e}$ the size of the set $A_{\mathfrak{e}}$.  Then analogously we define $\mathfrak{n} = t-\mathfrak{e}$ as the number of QFFs without an EC; $A_{\mathfrak{n}} = A \setminus A_{\mathfrak{e}}$ as the additional polynomials required to define them; and $m_{\mathfrak{n}} = m - m_{\mathfrak{e}}$ as their number.

Let $\mathcal{A}$ be the sequence of sets of polynomials $A_i$ defining each formula.  If QFF $i$ is one of the $\mathfrak{e}$ with an EC then set $E_i$ to be the set containing just that EC, and otherwise set $E_i = A_i$.  As before, denote the irreducible bases of these by $B_i$ and $F_i$.

\begin{lem}
\label{L:TTICAD2}
Under the assumptions above $P_{\mathcal{E}}(\mathcal{A})$ has the $(M,2d^2)$-property with  
\begin{equation}
\label{eq:TTIGeneral-M1}
M = \left\lfloor \tfrac{1}{2} (m_{\mathfrak{n}}+1)^2 \right\rfloor 
+ \left \lfloor \tfrac{1}{2}(3m_{\mathfrak{e}}+1) \right\rfloor 
+ \tfrac{1}{2}\mathfrak{e}(\mathfrak{e}-1 + 2m_{\mathfrak{n}} ).
\end{equation}
\end{lem}
\begin{proof}
Without loss of generality suppose the QFFs are labelled so the $\mathfrak{e}$ QFFs with an EC come first.  We will decompose the cross resultant set (\ref{eqn:RESX}) as $R^{\times}_1 \cup R^{\times}_2 \cup R^{\times}3$ 
where
\begin{align*}
R^{\times}_1 &= 
\{ {\rm res}_{x_n}(f,\hat{f}) \mid \exists i,j : \, f \in F_i, \hat{f} \in F_j, i<j\leq \mathfrak{e}, f \neq \hat{f}  \}, \\
R^{\times}_2 &= 
\{ {\rm res}_{x_n}(f,\hat{f}) \mid \exists i,j : \, f \in F_i, \hat{f} \in F_j, i\leq \mathfrak{e}<j, f \neq \hat{f}  \}, \\
R^{\times}_3 &= 
\{ {\rm res}_{x_n}(f,\hat{f}) \mid \exists i,j : \, f \in F_i, \hat{f} \in F_j, \mathfrak{e}<i<j, f \neq \hat{f}  \}.
\end{align*}
Then the projection set (\ref{eq:TTIProj}) may be decomposed as 
\begin{align}
P_{\mathcal{E}}(\mathcal{A}) &= \cont(\mathcal{A}) \cup \textstyle{\bigcup_{i=1}^t} P_{F_i}(B_i) \cup {\rm Res}^{\times} (\mathcal{F})  \nonumber \\
&= \left( \textstyle{ \bigcup_{i=1}^{\mathfrak{e}} } \cont(A_i) \cup P_{F_i}(B_i) \right) \nonumber \\
&\qquad 
\cup \left( R^{\times}_3 \cup \textstyle{ \bigcup_{i=\mathfrak{e}+1}^{t} } \cont(A_i) \cup P_{F_i}(B_i) \right)
\cup R^{\times}_1 \cup R^{\times}_2. \label{eq:TTIGenProof}
\end{align}
\begin{enumerate}
\item The first collection of sets in (\ref{eq:TTIGenProof}) has the 
$\left( \left\lfloor \tfrac{1}{2}(3m_{\mathfrak{e}}+1) \right\rfloor, 2d^2\right)$-property.  The argument is identical to the proof of Lemma \ref{L:TTICAD1}, except that here $\mathfrak{e}$ plays the role of $t$, and $m_{\mathfrak{e}}$ the role of $m$.
\item The second collection of sets in (\ref{eq:TTIGenProof}) refer to those with $E_i=A_i$.  Since $P_{B_i}(B_i) = P(B_i)$ we see that the union of $\cont(A_i) \cup P(B_i)$ for $i=\mathfrak{e}+1, \dots, t$ contains all the polynomials in $P(A_{\mathfrak{n}})$ except for the cross-resultants of polynomials from different $B_i$.  These are exactly given by $R^{\times}_3$, and thus we can follow the proof of Lemma \ref{L:SINew} to partition the second collection into 
$\left\lfloor \tfrac{1}{2}(m_{\mathfrak{n}}+1)^2 \right\rfloor$ sets of combined degree $2d^2$.
\item Next let us consider $R^{\times}_1$.  This concerns those subsets $E_i$ with only one polynomial, and hence their square free bases $F_i$ each have the $(1,d)$-property. Following the proof of Lemma \ref{L:SINew} step 2 this set of resultants may be partitioned into $\tfrac{1}{2}\mathfrak{e}(\mathfrak{e}-1)$ sets of combined degree at most $2d^2$.  
\item Finally we consider $R^{\times}_3$.  This concerns resultants of the $\mathfrak{e}$ polynomials forming the $\mathfrak{e}$ single polynomial subsets $E_i$, taken with polynomials from the other subsets (together giving the set $A_{\mathfrak{n}}$ of $m_{\mathfrak{n}}$ polynomials).  There are at most $\mathfrak{e}m_{\mathfrak{n}}$ of these.  Of course, as before, we are actually dealing with square free bases (moving from polynomials of degree $d$ to sets with the $(1,d)$-property) and then consider the coprime subsets (as in Lemma \ref{L:SINew}), to conclude $R^{\times}_3$ has the $(\mathfrak{e}m_{\mathfrak{n}}, 2d^2)$-property.
\end{enumerate}
Summing up then gives the desired result.
\end{proof}

\begin{cor}
\label{cor:TTIGen}
The bound in (\ref{eq:TTIGeneral-M1}) may be improved to
\begin{equation}
\label{eq:TTIGeneral-M}
M = \left\lfloor \tfrac{1}{2} \left( (m_{\mathfrak{n}}+1)^2 + 3m_{\mathfrak{e}} \right) \right\rfloor 
+ \tfrac{1}{2}  \left( \mathfrak{e}(\mathfrak{e}-1 + 2m_{\mathfrak{n}})\right).
\end{equation}
\end{cor}
\begin{proof}
We have asserted that the sum of the two floors is equal to the floor of the sum minus a half.  
In both steps 1 and 2 of the proof of Lemma \ref{L:TTICAD2} we pair up sets of maximum combined degree $d^2$ to get half as many with maximum combined degree $2d^2$.  We introduce the floor of the polynomial one greater to cover the case with an odd number of sets to begin with.  However, in the case that both step 1 and step 2 had an odd number of starting sets the left over couple could themselves be paired.  Instead, if we considering combining these sets and then pairing we have the floor as stated in (\ref{eq:TTIGeneral-M}).
\end{proof}

We analyse Algorithm \ref{alg:TTICAD} by applying Lemma \ref{L:TTICAD2} once and then Lemma \ref{L:SINew} repeatedly.  As usual, the growth is given by Table \ref{tab:GeneralProjection}, this time with $M$ as in (\ref{eq:TTIGeneral-M}). The dominant term in the bound on cell count is then calculated from (\ref{bound:All}) as
\begin{align*}
&\qquad 2^{2^{n}-1}d^{2^{n}-1}( \tfrac{1}{2} \left( (m_{\mathfrak{n}}+1)^2 + (3m_{\mathfrak{e}}+1) + \mathfrak{e}(\mathfrak{e}-1) + 2\mathfrak{e}m_{\mathfrak{n}} - 1 \right) )^{2^{n-1}-1}m \nonumber \\
&= 2^{2^{n-1}}d^{2^{n}-1}((m_{\mathfrak{n}}+1)^2 + (3m_{\mathfrak{e}}+1) + \mathfrak{e}(\mathfrak{e}-1) + 2\mathfrak{e}m_{\mathfrak{n}} -1 )^{2^{n-1}-1}m.
\end{align*}
Once again, we can improve this by noting the reduction at the final lift, which will involve $m_{\mathfrak{n}} + \mathfrak{e} \leq m$ polynomials instead of $m$.  Thus the bound becomes
\begin{equation}
\label{bound:TTI2}
2^{2^{n-1}}d^{2^{n}-1}((m_{\mathfrak{n}}+1)^2 + (3m_{\mathfrak{e}}+1) + \mathfrak{e}(\mathfrak{e}-1) + 2\mathfrak{e}m_{\mathfrak{n}} - 1)^{2^{n-1}-1}(m_{\mathfrak{n}} + \mathfrak{e}).
\end{equation}

\subsubsection*{Comparison}

First we consider three extreme cases for the TTICAD algorithm:
\begin{enumerate}
\item If no QFF has an EC then $\mathfrak{e}=0, m_{\mathfrak{e}}=0, m_{\mathfrak{n}}=m$ and (\ref{eq:TTIGeneral-M}) becomes $\lfloor \tfrac{1}{2}(m+1)^2 \rfloor$.  The latter is (\ref{eq:SI-M}) for sign-invariant CAD.

\item The other unfortunate case is when $\mathfrak{e} = m_{\mathfrak{e}}$, i.e. all those QFFs with an EC contain no other constraints.  In this case (\ref{eq:TTIGeneral-M}) becomes 
$\left\lfloor \tfrac{1}{2}(\mathfrak{e}+m_{\mathfrak{n}}+1)^2 \right\rfloor$, which will be equal to (\ref{eq:SI-M}) for sign-invariant CAD.  

\item The third extreme case is where all QFFs have an EC.  Then $\mathfrak{e}=t, m_{\mathfrak{e}}=m, m_{\mathfrak{n}}=0$ and (\ref{eq:TTIGeneral-M}) becomes 
$\left\lfloor\tfrac{1}{2}(3m+1)\right\rfloor + \tfrac{1}{2}(t-1)t$.  This is the same as (\ref{eq:TTI-M}) for the restricted case of TTICAD studied in Section \ref{subsec:CA-TTI1}.
\end{enumerate}
In all three cases the general TTICAD algorithm behaves identically to those previous approaches.  In the first two extreme cases the general TTICAD algorithm performs the same as \cite{McCallum1998} which produces a sign-invariant CAD.  Let us demonstrate that it is superior otherwise.  Assume $0 < \mathfrak{e} < m_{\mathfrak{e}}$ (meaning at least one QFF has an EC and at least one such QFF has additional constraints).  Then comparing the values of $M$ in (\ref{eq:SI-M}) and (\ref{eq:TTIGeneral-M}) we have:
\begin{align*}
M_{SI} - M_{TTI} &= 
\left\lfloor \tfrac{1}{2}
(m_{\mathfrak{e}}-\mathfrak{e})( \mathfrak{e}+m_{\mathfrak{e}} + 2m_{\mathfrak{n}}-1) 
\right\rfloor.
\end{align*}
The first factor is positive by assumption, and the second is $\geq 2$.  Thus the bound on the cell count for TTICAD is better than for sign-invariant CAD by at least a doubly exponential factor: $2^{2^{n-1}-1}$.

There is no need to compare the complexity for TTICAD in this general case to any use of \cite{McCallum1999}.  The latter can only be applied to a parent formula with an overall (possibly implicit) EC and the construction from the previous subsection would only be possible when $\mathfrak{e}=t$: the case of the previous subsection for which we have already concluded the superiority of TTICAD.

It is now clear that the extension to general QFFs provided by this paper is a more important contribution than the restricted case of \cite{BDEMW13}, even though the former has a lower complexity bound:
\begin{itemize}
\item In the restricted case TTICAD was an improvement on the best available alternative projection operator, $P_E(A)$ from \cite{McCallum1999}, but its improvements were to the base of a double exponential factor.
\item Outside of this restricted case (and the two other extreme cases) TTICAD offers a complexity improvement to a double exponent when compared with the best available alternative projection operator, $P(A)$ from \cite{McCallum1998}.
\end{itemize}

\section{Our implementation in Maple}
\label{sec:Implementation}

There are various implementations of CAD already available including:
{\sc Mathematica} \citep{Strzebonski2006, Strzebonski2010};
{\sc Qepcad} \citep{Brown2003a};
the \texttt{Redlog} package for \textsc{Reduce} \citep{SS03};
the \texttt{RegularChains} Library \citep{CMXY09} for {\sc Maple},
and \texttt{SyNRAC} \citep{YA06} (another package for {\sc Maple}).

None of these can (currently) be used to build CADs which guarantee order-invariance, a property required for proving the correctness of our TTICAD algorithm.
Hence we have built our own CAD implementation in order to obtain experimental results for our ideas.

\subsection{ProjectionCAD}
\label{subsec:PCAD}

Our implementation is a third party {\sc Maple} package which we call \texttt{ProjectionCAD}.  It gathers together algorithms for producing CADs via projection and lifting to complement the CAD commands which ship with \textsc{Maple} and use the alternative approach based on the theory of regular chains and triangular decomposition.

All the projection operators discussed in Sections \ref{sec:ExistingProj} and \ref{sec:TTIProj} have been implemented and so \texttt{ProjectionCAD} can produce CADs which are sign-invariant, order-invariant, invariant with respect to a declared EC, and truth table invariant.  Stack generation (step \ref{step:lifting2} in Algorithm \ref{alg:TTICAD}) is achieved using an existing command from the \texttt{RegularChains} package, described fully in Section 5.2 of \citep{CMXY09}.  To use this we must first process the input to satisfy the assumptions of that algorithm: that polynomials are co-prime and square-free when evaluated on the cell (\textit{separate above the cell} in the language of regular chains).  This is achieved using other commands from the \texttt{RegularChains} library.

Utilising the \texttt{RegularChains} code like this means that \texttt{ProjectionCAD} can represent and present CADs in the same way.  In particular this allows for easy comparison of CADs from the different implementations; the use of existing tools for studying the CADs; and the ability to display CADs to the user in the easy to understand \texttt{piecewise} representation \citep{CDMMXXX09}.
Figure \ref{fig:Maple} shows an example of the package in use.  

\begin{figure}[ht]
\caption{An example of using \texttt{ProjectionCAD} to build a sign-invariant CAD for the unit circle.  The output is as displayed in {\sc Maple}, but with sample points replaced by $SP$ for brevity.}
\label{fig:Maple}
\begin{framed}
\begin{verbatim}
> f := x^2+y^2-1:
\end{verbatim}
\begin{verbatim}
> cad := CADFull([f], vars, method=McCallum, output=piecewise);
\end{verbatim}
\[
\begin{cases}
SP & \quad x<-1 \\
\begin{cases} SP & \quad y<0 \\ SP & \quad y=0 \\ SP & \quad 0<y \end{cases}  & \quad x=-1 \\
\begin{cases}
SP & \qquad\qquad y<-\sqrt{-{x}^{2}+1} \\
SP & \qquad\qquad y=-\sqrt {-{x}^{2}+1} \\
SP & {\it And} \left( -\sqrt {-{x}^{2}+1}<y,y<\sqrt {-{x}^{2}+1} \right) \\
SP & \qquad\qquad y=+\sqrt {-{x}^{2}+1} \\
SP & \qquad\qquad \sqrt {-{x}^{2}+1}<y
\end{cases}
&{\it And} \left( -1<x,x<1 \right) \\
\begin{cases} SP & \quad y<0 \\ SP & \quad y=0 \\ SP & \quad 0<y \end{cases}  & \quad x=1  \\
SP & \quad 1<x
\end{cases}
\]
\begin{verbatim}
> CADNumCellsInPiecewise(cad);
\end{verbatim}
\[
13
\]
\end{framed}
\end{figure}

Unlike {\sc Qepcad}, \texttt{ProjectionCAD} has an implementation of delineating polynomials (actually the minimal delineating polynomials of \cite{Brown2005a}) and so it can solve certain problems without unnecessary warnings.  It is also the only  
CAD implementation that can reproduce the theoretical algorithm {\tt CADW}.

Other notable features of \texttt{ProjectionCAD} include commands to present the different formulations of problems for the algorithms and heuristics to help choose between these.
For more details on \texttt{ProjectionCAD} and the algorithms implemented within see  \citep{EWBD14}, while the package itself is freely available from the authors along with documentation and examples demonstrating the functionality.  To run the code users need a version of {\sc Maple} and the \texttt{RegularChains} Library.

\subsection{Minimising failure of TTICAD}
\label{subsec:Excl}

Algorithm \ref{alg:TTICAD} was kept simple to aid readability and understanding.  Our implementation does make some extra refinements.  Most of these are trivial, such as removing constants from the set of projection polynomials or when taking coefficients in order of degree, stopping if the ones already included can be shown not to vanish simultaneously.

The well-orientedness conditions can often be overly cautious.  \cite{Brown2005a} discussed cases where non-well oriented input can still lead to an order-invariant CAD. Similarly here, we can sometimes allow the nullification of an EC on a positive dimensional cell.
Define the {\em excluded projection polynomials} for each $i$ as: 
\begin{align}
{\rm ExclP}_{E_i}(A_i) &:= P(A_i) \setminus P_{E_i}(A_i)
\label{eq:Excl} \\
&= \{ {\rm coeffs}(g), {\rm disc}_{x_n}(g), {\rm res}_{x_n}(g,\hat{g}) \mid g, \hat{g} \in A_i \setminus E_i,  g \neq \hat{g} \}.\nonumber
\end{align}
Note that the total set of excluded polynomials from $P(A)$ will include all the entries of the ${\rm ExclP}_{E_i}(A_i)$ as well as missing cross resultants of polynomials in $A_i \setminus E_i$ with polynomials from $A_j \neq A_i$.  

\begin{lem}
\label{lem:ConstPolys}
Let $f_i$ be an EC which vanishes identically on a cell $c \in \mathcal{D}'$ constructed during Algorithm \ref{alg:TTICAD}. If all polynomials in ${\rm ExclP}_{E_i}(A_i)$ are constant on $c$ then any $g \in A_i \setminus E_i$ will be delineable over $c$.
\end{lem}

\begin{proof}
Suppose first that $A_i$ and $E_i$ satisfy the simplifying conditions from Section \ref{subsec:TTIProjOp}. Rearranging \eqref{eq:Excl} we see
$P(A_i) = P_{E_i}(A_i) \cup {\rm ExclP}_{E_i}(A_i)$.
However, given the conditions of the lemma, this is equivalent (after the removal of constants which do not affect CAD construction) to $P_{E_i}(A_i)$ on $c$.  So here $P(A_i)$ is a subset of $P_{\mathcal{E}}(\mathcal{A})$ and we can conclude by Theorem \ref{thm:McC1} that all elements of $A_i$ vanish identically on $c$ or are delineable over $c$.

We can draw the same conclusion in the more general case of $A_i$ and $E_i$ because $P(A_i) = C_i \cup P_{F_i}(B_i) \cup {\rm ExclP}_{F_i}(B_i) \subseteq \mathfrak{P}$.
\end{proof}

Hence Lemma \ref{lem:ConstPolys} allows us to extend Algorithm \ref{alg:TTICAD} to deal safely with such cases.  Although we cannot conclude sign-invariance we can conclude delineability and so instead of returning failure we can proceed by extending the lifting set $L_c$ to the full set of polynomials (similar to the case of nullification on a cell of dimension zero dealt with in step \ref{step:addthebi} of Algorithm \ref{alg:TTICAD}).
In particular, this allows for ECs $f_i$ which do not have main variable $x_n$.  Our implementation makes use of this.

Note that the widening of the lifting step here (and also in the case of the zero dimensional cell) is for the generation of the stack over a single cell.  The extension is only performed for the necessary cells thus minimising the cell count while maximising the success of the algorithm, as shown in Example \ref{ex:ExtendingLiftingSet}.  Since a polynomial cannot be nullified everywhere such case distinction will certainly decrease the amount of lifting.

\begin{exmp}
\label{ex:ExtendingLiftingSet}
Consider the polynomials
\[
f = z+yw, \quad
g = yx+1, \quad
h = w(z+1)+1,
\]
the single formula $f=0 \wedge g<0 \wedge h<0$ and assume the variable ordering
$x \prec y \prec z \prec w$.
Using the \texttt{ProjectionCAD} package we can build
a TTICAD with 467 cells for this formula.
The induced CAD of $\R^3$, $D$,  has 169 cells and on five of these cells the polynomial $f$ is nullified.  On these five cells both $y$ and $z$ are zero, with $x$ being either fixed to $0,4$ or belonging to the three intervals splitting at these points.

In this example ExclP$_E(A) = \{z+1\}$ arising from the coefficient of $h$.  This is a constant value of 1 on all five of those cells.  Thus the algorithm is allowed to proceed without error, lifting with respect to all the projection polynomials on these cells.

The lifting set varies from cell to cell in $D$.  For example, the stack over the cell $c_1 \in D$ where $x=y=z=0$ uses three cells, splitting when $w=-1$.  This is required for a CAD invariant with respect to $f$ since $f=0$ on $c$ but $h$ changes sign when $w=-1$.  Compare this with, for example, the cell $c_2 \in D$ where $x=y=0$ and $z<-1$.  The stack over $c_2$ has only one cell, with $w$ free.  The polynomial $h$ will change sign over this cell, but this is not relevant since $f$ will never be zero.  This occurs because $h$ is included in the lifting set only for the five cells of $D$ where $f$ was nullified.
\end{exmp}

In  theory, we could go further and allow this extension to apply when the polynomials in ${\rm ExclP}_{E_i}(A_i)$ are not necessarily all constant, but have no real roots within the cell $c$. However, identifying such cases would, in general, require answering a separate quantifier elimination question, which may not be trivial, and so this has yet to be implemented.

\subsection{Formulating problems for TTICAD}
\label{sec:Formulation}

When using Algorithm \ref{alg:TTICAD} various choices may be required which can have significant effects on the output.  We briefly discuss some of these possibilities here.

\subsubsection{Variable ordering}
\label{subsec:VarOrd}

Algorithm \ref{alg:TTICAD} runs with an ordering on the variables.  As with all CAD algorithms this ordering can have a large effect, even determining whether a computation is feasible.  \cite{BD07} presented problem classes where one ordering gives a constant cell count, and another a cell count doubly exponential in the number of variables.  

Some of the ordering may already be determined.  For example, when using a CAD for quantifier elimination the quantified variables must be eliminated first.  However, even then we are free to change the ordering of the free variables, or those in quantifier blocks.  Various heuristics have been developed to help with this choice:
\begin{description}
\item[\cite{Brown2004}:] 
Choose the next variable to eliminate according to the following criteria on the input, starting with the first and breaking ties with successive ones: 
\begin{enumerate}[(1)]
\item lowest overall degree in the input with respect to the variable;
\item lowest (maximum) total degree of those terms in the input in which it occurs;
\item smallest number of terms in the input which contain the variable.
\end{enumerate}
\item[sotd \citep{DSS04}:] 
Construct the full set of projection polynomials for each ordering and select the ordering whose set has the lowest \emph{sum of total degree} for each of the monomials in each of the polynomials. 
\item[ndrr \citep{BDEW13}:] 
Construct the full projection set and select the one with the lowest \emph{number of distinct real roots} of the univariate polynomials.
\item[fdc \citep{WEBD14}:] Construct all full-dimensional cells for different orderings (requires no algebraic number computations) and select the smallest.
\end{description}   
The Brown heuristic perform well despite being low cost.  A machine learning experiment by \cite{HEWDPB14} showed that each heuristic had classes of examples where it was superior, and that a machine learned choice of heuristic can perform better than any one.

\begin{exmp}
\label{ex:Kahan}
\cite{Kahan87} gives a classic example for algebraic simplification in the presence of branch cuts.  He considers a fluid mechanics problem leading to the relation
\begin{equation}
\label{eq:Kahan}
2\rm{arccosh}\left(\frac{3+2z}3\right)-\rm{arccosh}\left(\frac{5z+12}{3(z+4)}\right)= 2\rm{arccosh}\left(2(z+3)\sqrt{\frac{z+3}{27(z+4)}}\right).
\end{equation}
This is true over all $\mathbb{C}$ except for the small teardrop region shown on the left of Figure \ref{fig:Kahan}: a plot of the imaginary part of the difference between the two sides of (\ref{eq:Kahan}).

Recent work described in \citep{EBDW13} allows for the systematic identification of semi-algebraic formula to describe branch cuts.  This, along with visualisation techniques, now forms part of \textsc{Maple}'s \texttt{FunctionAdvisor} \citep{EC-TBDW14}.  For this example the technology produces the plot on the right of Figure \ref{fig:Kahan} and describes the branch cuts using 7 pairs of equations and inequalities.
With \texttt{ProjectionCAD}, a  sign-invariant CAD for these polynomials has 409 cells using $x \prec y$ and 1143 with $y \prec x$, while a TTICAD has 55 cells using $x \prec y$ and 39 with $y \prec x$.  

So the best choice of variable ordering differs depending on the CAD algorithm used.  For the sign-invariant CAD, all three heuristics described above identify the correct ordering, so it would have been best to use the cheapest, \texttt{Brown}.  However, for the TTICAD only the more expensive \texttt{ndrr} heuristic selects the correct ordering.
\end{exmp}

\begin{figure}[ht]
\begin{center}
\includegraphics[width=4.5cm]{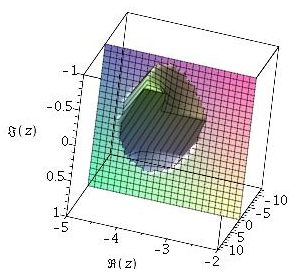}
\hspace*{0.3cm}
\includegraphics[width=4.5cm]{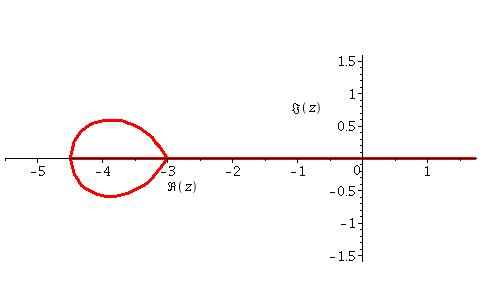}
\end{center}
\caption{Plots relating to equation (\ref{eq:Kahan}) from Example \ref{ex:Kahan}.}
\label{fig:Kahan}
\end{figure}

\subsubsection{Equational constraint designation and logical formulation}
\label{subsec:LogicalForm}

If any QFF has more than one EC present then we must choose which to designate for speical use in Algorithm \ref{alg:TTICAD}.  As with the variable ordering choice, this leads to two different projection sets which could be compared using the $\texttt{sotd}$ and $\texttt{ndrr}$ measures. 

However, note that this situation actually offers more choice than just the designation.  If $\phi_i$ had two ECs then it would be admissible to split it into two QFFs $\phi_{i,1}, \phi_{i,2}$ with one EC assigned to each and the other constraints partitioned between them in any manner.  Admissible because any TTICAD for $\phi_{i,1}, \phi_{i,2}$ is also a TTICAD for $\phi_i$.

This is a generalisation of the following observation: given a formula $\phi$ with two ECs a CAD could be constructed using either the original theory of \cite{McCallum1999} or the  TTICAD algorithm applied to two QFFs.  The latter option would certainly lead to more projection polynomials.  However, a specific EC may have a comparatively large number of intersections with another constraint, in which case, separating them into different QFFs could still offer benefits (with the increase in projection polynomials offset by them having less real roots).  The following is an example of such a situation.

\begin{exmp}
\label{ex:TTIorNot}
Assume $x \prec y$ and consider again
$\Phi := (f_1 = 0 \wedge g_1>0) \vee (f_2 = 0 \wedge g_2<0)$
but this time with polynomials below.  These are plotted in Figure \ref{fig:FormEx} where the solid curve is $f_1$, the solid line $g_1$, the dashed curve $f_2$ and the dashed line $g_2$.
\begin{align*}
f_1 &:= (y-1) - x^3+x^2+x, \qquad \quad
g_1 := y - \tfrac{x}{4}+\tfrac{1}{2}, \\
f_2 &:= (-y-1) - x^3+x^2+x, \qquad \,
g_2 := -y - \tfrac{x}{4}+\tfrac{1}{2},
\end{align*}
If we use the algorithm by \cite{McCallum1999} with the implicit EC $f_1f_2=0$ designated then a CAD is constructed which identifies all the intersections except for $g_1$ with $g_2$
This is visualised by the plot on the left while the plot on the right relates to a TTICAD with two QFFs.  In this case only three 0-cells are identified, with the intersections of $g_2$ with $f_1$ and $g_1$ with $f_2$ ignored.
The TTICAD has 31 cells, compared to 39 cells for the other two.
Both \texttt{sotd} and \texttt{ndrr} identify the smaller CAD, while \texttt{Brown} would not discriminate.
\end{exmp}

\begin{figure}[ht]
\begin{center}
\includegraphics[width=6.0cm]{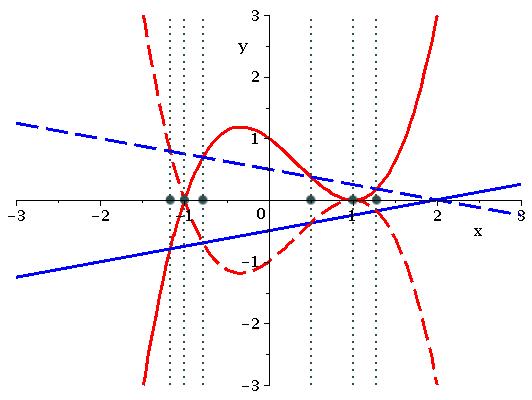}
\hspace*{0.3cm}
\includegraphics[width=6.0cm]{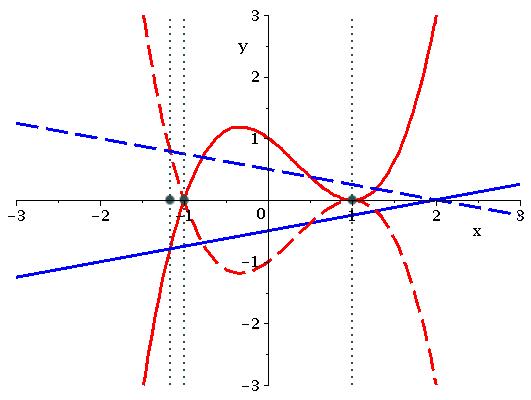}
\end{center}
\caption{Plots visualising the CADs described for Example \ref{ex:TTIorNot}.}
\label{fig:FormEx}
\end{figure}


More details on the issues around the logical formulation of problems for TTICAD is given by \cite{BDEW13}.

\subsubsection{Preconditioning input QFFs}
\label{subsec:Grobner}

Another option available before using Algorithm \ref{alg:TTICAD} is to precondition the input.  \cite{BH91} conducted experiments to see if Gr\"obner basis techniques could help CAD.  They considered replacing any input polynomials which came from equations by a purely lexicographical Gr\"obner basis for them.
In \citep{WBD12_GB} this idea was investigated further with a larger base of problems tested and the idea extended to include Gr\"obner reduction on the other polynomials.
The preconditioning was shown to be highly beneficial in some cases, but detrimental in others.  A simple metric was posited and shown to be a good indicator of when preconditioning was useful

\cite{BDEW13} consider using Gr\"obner preconditioning for TTICAD by constructing bases for each QFF. 
This can produce significant reductions in the TTICAD cell counts and timings.  The benefits are not universal, but measuring the \texttt{sotd} and \texttt{ndrr} of the projection polynomials gives suitable heuristics.

\subsubsection{Summary}
\label{subsec:FormulationSummary}

We have highlighted choices we may need to make before using Algorithm \ref{alg:TTICAD} and its implementation in \texttt{ProjectionCAD}.  The heuristics discussed are also available in that package.  An issue of problem formulation not described in the mathematical derivation of the problem itself.  We note that this can have a great effect on the tractability of using CAD (see \cite{WDEB13} for example).

For the experimental results in Section \ref{sec:Experiment} we use the specified variable ordering for a problem if it has one and otherwise test all possible orderings.  If there are questions of logical formulation or EC designation we use the heuristics discussed here.  No Gr\"obner preconditioning was used as the aim is to analyse the TTICAD theory itself.

It is important to note that the heuristics are just that, and as such can be misled by certain examples. Also, while we have considered these issues individually they of course intersect.  For example, the TTICAD formulation with two QFFs was the best choice in Example \ref{ex:TTIorNot} but if we had assumed the other variable ordering then a single QFF is superior.
Taken together, all these choices of formulation can become combinatorially overwhelming and so methods to reduce this, such as the greedy algorithm in \citep{DSS04} or the suggestion in Section 4 of \cite{BDEW13} are important.

\section{Experimental Results}
\label{sec:Experiment}

\subsection{Description of experiments}
\label{subsec:ERDescription}

Our timings were obtained on a Linux desktop (3.1GHz Intel processor, 8.0Gb total memory) with {\sc Maple} 16 (command line interface), {\sc Mathematica} 9 (graphical interface) and {\sc Qepcad-B} 1.69.
For each experiment we produce a CAD and give the time taken and cell count. The first is an obvious metric while the second is crucial for applications performing operations on each cell.

For {\sc Qepcad} the options {\tt +N500000000} and {\tt +L200000} were provided, the initialization included in the timings and ECs declared when possible (when they are explicit or formed by the product of ECs for the individual QFFs).
In {\sc Mathematica} the output is not a CAD but a formula constructed from one \citep{Strzebonski2010}, with the actual CAD not available to the user.  Cell counts for the algorithms were provided by the author of the {\sc Mathematica} code.

TTICADs are calculated using our \texttt{ProjectionCAD} implementation described in Section \ref{sec:Implementation}.  The results in this section are not presented to claim that our implementation is state of the art, but to demonstrate the power of the TTICAD theory over the conventional theory, and how it can allow even a simple implementation to compete.  Hence the cell counts are of most interest.

The time is measured to the nearest tenth of a second, with a time out ({\bf T}) set at $5000$ seconds.   When {\bf F} occurs it indicates failure due to a theoretical reason such as not well-oriented (in either sense).  The occurrence  of Err indicates an error in 
an internal subroutine of {\sc Maple}'s \texttt{RegularChains} package, used by \texttt{ProjectionCAD}.  This error is not theoretical but a bug, which will be fixed shortly.

We started by considering examples originating from \citep{BH91}.  However these problems (and most others in the literature) involve conjunctions of conditions, chosen as such to make them amenable to existing technologies.
These problems can be tackled using TTICAD, but they do not demonstrate its full strength. Hence we introduce some new examples.  The first set, those denoted with a $\dagger$, are adapted from \citep{BH91} by turning certain conjunctions into disjunctions.  The second set were generated randomly as examples with two QFFs, only one of which has an EC (using random polynomials in 3 variables of degree at most 2).

Two further examples came from the application of branch cut analysis for simplification.
We included Example \ref{ex:Kahan} along with the problem induced by considering the validity of the double angle formulae for arcsin.  Finally we considered the worked examples from Section \ref{subsec:WE1} and the generalisation to three dimensions presented in Example \ref{ex:3dPhi}.
Note that A and B following the problem name indicate different variable orderings.  Full details for all examples can be found in the CAD repository \citep{WBD12_EX}
available freely at \texttt{http://dx.doi.org/10.15125/BATH-00069}.

\subsection{Results}
\label{subsec:ERResults}

We present our results in Table \ref{table:Results}.  For each problem we give the name used in the repository, $n$ the number of variables, $d$ the maximum degree of polynomials involved and $t$ the number of QFFs used for TTICAD.  We then give the time taken (T) and number of cells of $\mathbb{R}^n$ produced (C) by each algorithm.

\begin{table}[b]
\caption{Comparing TTICAD to other CAD types and other CAD implementations.}
\label{table:Results}
\begin{center}
\begin{footnotesize}
\begin{tabular}{lc|rrrrrrrrrr}
  \multicolumn{2}{c|}{Problem}
& \multicolumn{2}{c}{Full-CAD} & \multicolumn{2}{c}{TTICAD}
& \multicolumn{2}{c}{{\sc Qepcad}} & \multicolumn{2}{c}{{\sc Maple}}
& \multicolumn{2}{c}{{\sc Mathematica}} \\
Name & n\,d\,t
                        &  T \,    &  C \,     &  T \,     &  C \,
                        &  T \,    &  C \,     &  T \,     &  C \,
                        &  T \,    &  C \,  \\
\hline
IntA                    & 3\,2\,1
						& 360     &  3707       &  1.7       &  269
						& 4.5       &  825        &  ---       &  Err
						& 0.0 & 3   \\
IntB                    & 3\,2\,1
						& 332     &  2985       &  1.5       &  303
						& 4.5       &  803        &  50.2      &  2795
						& 0.0 & 3   \\
RanA                   & 3\,3\,1
						& 269     &  2093       &  4.5       &  435
						& 4.6       &  1667       &  23.0      &  1267
						& 0.1 & 657 \\
RanB                   & 3\,3\,1
						& 443     &  4097       &  8.1       &  711
						& 5.4       &  2857       &  48.1      &  1517
						& 0.0 & 191 \\
Int$\dagger$A           & 3\,2\,2
						&  360     &  3707       &  68.7      &  575
						&  4.8       &  3723       &  ---       &  Err
						& 0.1 & 601 \\
Int$\dagger$A           & 3\,2\,2
						&  332     &  2985       &  70.0      &  601
						&  4.7       &  3001       &  50.2      &  2795
						& 0.1 & 549 \\
Ran$\dagger$A          & 3\,3\,2
						&  269     &  2093       &  223     &  663
						&  4.6       &  2101       &  23.0      &  1267
						& 0.2 & 808 \\
Ran$\dagger$B          & 3\,3\,2
						&  443     &  4097       &  268     &  1075
						&  142     &  4105       &  48.1      &  1517
						& 0.2 & 1156 \\
Ell$\dagger$A         & 5\,4\,2
						&  ---       &  {\bf F}    &  ---       &  {\bf F}
						&  292     &  500609
						     &  1940    &  81193
						& 11.2 & 80111 \\
Ell$\dagger$B         & 5\,4\,2
						&  \textbf{T}       &  ---        &  \textbf{T}       &  ---
						&  \textbf{T}       &  ---        &  \textbf{T}       &  ---
						& 2911
						&  $\genfrac{}{}{0pt}{1}{16,603,}{\quad 131}$ 
						\\ 
Solo$\dagger$A          & 4\,3\,2
						&  678     &  54037      &  46.1      &  {\bf F}
						&  4.9       &  20307      &  1014    &  54037
						& 0.1 & 260 \\
Solo$\dagger$B          & 4\,3\,2
						&  2009    &  154527 
						     &  123     &  {\bf F}
						&  6.3       &  87469      &  2952    
& 154527
						& 0.1 & 762 \\
Coll$\dagger$A          & 4\,4\,2
						&  265     &  8387       &  267     &  8387
						&  5.0       &  7813       &  376     &  7895
						& 3.6 & 7171 \\
Coll$\dagger$B          & 4\,4\,2
						&  ---       &  Err        &  ---         &  Err
						& \textbf{T}        &  ---        &  \textbf{T}       &  ---
						& 592 
						& $\genfrac{}{}{0pt}{1}{1,234,}{\quad 601}$ 
						\\
Ex\ref{ex:Kahan}A                    & 2\,4\,7
						&  10.7      &  409        &  0.3       &  55
						&  4.8       &  261        &  15.2      &  409
						& 0.0 & 72 \\
Ex\ref{ex:Kahan}B                    & 2\,4\,7
						&  87.9      &  1143       &  0.3       &  39
						&  4.8       &  1143       &  154     &  1143
						& 0.1 & 278 \\
AsinA                   & 2\,4\,4
						&  2.5       &  225        &  0.3       &  57
						&  4.6       &  225        &  3.3       &  225
						& 0.0 & 175 \\
AsinB                   & 2\,4\,4
						&  6.5       &  393        &  0.2       &  25
						&  4.5       &  393        &  7.8       &  393
						& 0.0 & 79 \\
Ex$\Phi$A           & 2\,2\,2
                        &  5.7       &  317        &  1.2       &  105
                        &  4.7       &  249        &  6.3       &  317
                        & 0.0 & 24 \\
Ex$\Phi$B           & 2\,2\,2
                        &  6.1       &  377        &  1.5       &  153
                        &  4.5       &  329        &  7.2       &  377
                        & 0.0 & 175 \\
Ex$\Psi$A           & 2\,2\,2
                        &  5.7       &  317        &  1.6       &  183
                        &  4.9       &  317        &  6.3       &  317
                        &  0.1 &  372  \\
Ex$\Psi$B           & 2\,2\,2
                        &  6.1       &  377        &  1.9       &  233
                        &  4.8       &  377        &  7.2       &  377
                        &  0.1   &  596   \\
Ex\ref{ex:3dPhi}A       & 3\,3\,2
						&  3796    &  5453       &  5.0       &  109
						&  5.3       &  739        &  ---       &  Err
						& 0.1 & 44 \\
Ex\ref{ex:3dPhi}B       & 3\,3\,2
						&  3405    &  6413       &  5.8       &  153
						&  5.7       &  1009       &  ---       &  Err
			 			& 0.1 & 135 \\
Rand1                   & 3\,2\,2
                        &  16.4      &  1533       &  76.8      &  1533
                        &  4.9       &  1535       &  25.7      &  1535
                        &  0.2      &  579        \\
Rand2                   & 3\,2\,2
                        &  838     &  7991       &  132     &  2911
                        &  5.2       &  8023       &  173     &  8023
                        &  0.8      &  2551       \\
Rand3                   & 3\,2\,2
                        &  259     &  8889       &  98.1      &  4005
                        &  5.3       &  8913       &  77.9      &  5061
                        &  0.7      &  3815       \\
Rand4                   & 3\,2\,2
                        &  1442    &  11979      &  167     &  4035
                        &  5.4       &  12031      &  258     &  12031
                        &  1.3       &  4339       \\
Rand5                   & 3\,2\,2
                        &  310     &  11869      &  110     &  4905
                        &  5.5       &  11893      &  104     &  6241
                        &  0.9       &  5041       \\
\end{tabular}
\end{footnotesize}
\end{center}
\end{table}

We first compare our TTICAD implementation with the sign-invariant CAD generated using \texttt{ProjectionCAD} with McCallum's projection operator.  Since these use the same architecture the comparison makes clear the benefits of the TTICAD theory.  The experiments confirm the fact that, since each cell of a TTICAD is a superset of cells from a sign-invariant CAD, the cell count for TTICAD will always be less than or equal to that of a sign-invariant CAD produced using the same implementation.  Ellipse$\dagger$ A is not well-oriented in the sense of \citep{McCallum1998}, and so both methods return {\bf FAIL}.  Solotareff$\dagger$ A and B are well-oriented in this sense but not in the stronger sense of Definition \ref{def:WO} and hence TTICAD fails while the sign-invariant CADs can be produced.  The only example with equal cell counts is Collision$\dagger$ A in which the non-ECs were so simple that the projection polynomials were unchanged.  Examining the results for the worked examples and the 3d generalisation we start to see the true power of TTICAD. In 3D Example A we see a 759-fold reduction in time and a 50-fold reduction in cell count.

We next compare our implementation of TTICAD with the state of the art in CAD: {\sc Qepcad} \citep{Brown2003a}, {\sc Maple} \citep{CMXY09} and {\sc Mathematica} \citep{Strzebonski2006, Strzebonski2010}.  {\sc Mathematica} is the quickest, however TTICAD often produces fewer cells.  We note that {\sc Mathematica}'s algorithm uses powerful heuristics and so actually used Gr\"obner bases on the first two problems, causing the cell counts to be so low.  When all implementations succeed TTICAD usually produces far fewer cells than {\sc Qepcad} or {\sc Maple}, especially impressive given {\sc Qepcad} is producing partial CADs for the quantified problems, while TTICAD is only working with the polynomials involved.

Reasons for the TTICAD implementation struggling to compete on speed may be that the {\sc Mathematica} and {\sc Qepcad} algorithms are implemented directly in {\tt C}, have had more optimization, and in the case of {\sc Mathematica} use validated numerics for lifting \citep{Strzebonski2006}.  However, the strong performance in cell counts is very encouraging, both due its importance for applications where CAD is part of a wider algorithm (such as branch cut analysis) and for the potential if TTICAD theory were implemented elsewhere.

\subsection{The increased benefit of TTICAD}
\label{subsec:IncreasedBenefit}

We finish by demonstrating that the benefit of TTICAD over the existing theory should increase with the number of QFFs and that this benefit is much more pronounced if at least one of these does not have an EC.

\begin{exmp}
\label{ex:IncreasedBenefit}
We consider a family of examples (to which our worked examples belong).
Assume $x \prec y$ and for $j$ a non-negative integer define
\begin{align*}
f_{j+1} &:= (x-4j)^2 + (y-j)^2 - 1, \qquad
g_{j+1} := (x-4j)(y-j) - \tfrac{1}{4}, \\
F_{j+1} &:= \{f_k, g_k\}_{k=1 \dots j+1}, \hspace*{0.7in}  
\textstyle \Phi_{j+1} := \bigvee_{k=1}^{j+1} ( f_{k}=0 \land g_{k}<0 ), \\
\Psi_{j+1} &:= \textstyle \left( \bigvee_{k=1}^{j} ( f_{k}=0 \land g_{k}<0 ) \right) \lor (f_{j+1}<0 \land g_{j+1}<0).
\end{align*}

Then $\Phi_2$ is $\Phi$ from equation (\ref{eqn:ExPhi}) and $\Psi_2$ is $\Psi$ from equation (\ref{eqn:ExPsi}).  Table \ref{tab:IncreasedBenefit} shows the cell counts for various CADs produced for studying the truth of the formulae, and Figure \ref{fig:IB} plots these values.

Both $\Phi_i$ and $\Psi_i$ may be studied by a sign-invariant CAD for the polynomials $F_i$, shown in the column marked \texttt{CADFull}.  The remaining CADs are specific to one formula.
For each formula a TTICAD has been constructed using Algorithm \ref{alg:TTICAD} on the natural sub-formulae created by the disjunctions, while the $\Phi_i$ have also had a CAD constructed using the theory of ECs alone.  This was simulated by running Algorithm \ref{alg:TTICAD} on the single formula declaring the product of the $f_i$s as an EC (column marked \texttt{ECCAD}).  All the proceeding CADs were constructed with \texttt{ProjectionCAD}.  For each formula a CAD has also been created with {\sc Qepcad}, with the product of $f_i$ declared as an EC for $\Phi_i$.

We see that the size of a sign-invariant CAD is grows much faster than the size of a TTICAD.  For a problem with fixed variable ordering the TTICAD theory seems to allow for linear growth in the number of formulae. Considering the \texttt{ECCAD} and \textsc{Qepcad} results shows that when all QFFs have an EC (the $\Phi_i$) using the implicit EC also makes significant savings.  However, it is only when using the improved lifting discussed in Section \ref{sec:ImprovedLifting} that these savings restrict the output to linear growth.
In the case where at least one QFF does not have an EC (the $\Psi_i$) the existing theory of ECs cannot be used.  So while the comparative benefit of TTICAD over sign-invariant CAD is slightly less, the benefit when comparing with the best available previous theory is far greater.
\end{exmp}

\begin{table}[b]
\caption{Table detailing the number of cells in CADs constructed to analyse the truth of the formulae from Example \ref{ex:IncreasedBenefit}.}
\label{tab:IncreasedBenefit}
\centering
\begin{tabular}{c|ccc|c|cc}
\textbf{j} & \multicolumn{3}{c|}{$\Phi_j$} & \textbf{$F_j$} & \multicolumn{2}{c}{$\Psi_j$} \\
           & \texttt{ECCAD} & \texttt{TTICAD} & \textsc{Qepcad} & \texttt{CADFull}
           & \texttt{TTICAD}  & \textsc{Qepcad} \\
\hline
2 & 145 & 105 & 249  & 317  & 183 & 317   \\
3 & 237 & 157 & 509  & 695  & 259 & 695   \\
4 & 329 & 209 & 849  & 1241 & 335 & 1241  \\
5 & 421 & 261 & 1269 & 1979 & 411 & 1979  \\
6 & 513 & 313 & 1769 & 2933 & 487 & 2933  \\
\end{tabular}
\end{table}

\begin{figure}[b]
\caption{Plots of the results from Table \ref{tab:IncreasedBenefit}.   The $x$-axis measures $j$ and the $y$-axis the number of cells.  On the left are the algorithms relating to $\Phi_j$ which from top to bottom are: \texttt{CADFull}, \textsc{Qepcad}, \texttt{ECCAD}, \texttt{TTICAD}.  On the right are the algorithms relating to $\Psi_j$ which from top to bottom are: \texttt{CADFull} and \texttt{TTICAD}. }
\label{fig:IB}
\begin{center}
\includegraphics[width=2.5in]{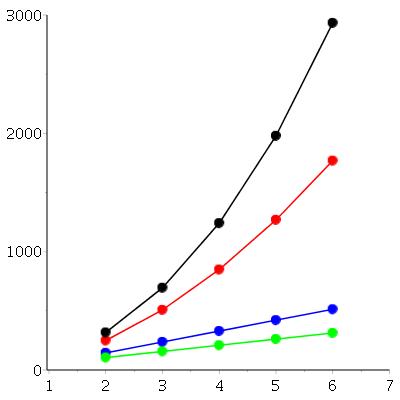}
\includegraphics[width=2.5in]{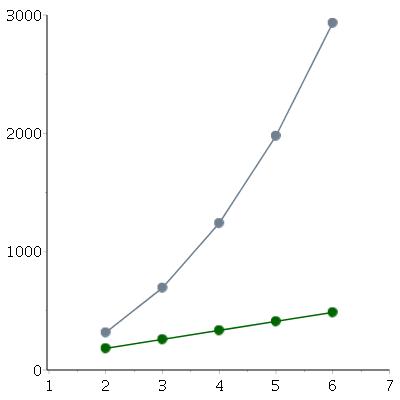}
\end{center}
\end{figure}

\section{Conclusions}
\label{sec:Conclusion}

We have defined truth table invariant CADs and by building on the theory of equational constrains have provided an algorithm to construct these efficiently.  We have extended the our initial work in ISSAC 2013 so that it applies to a general sequence of formulae.  The new complexity analyses show that the benefit over previously applicable CAD projection operators is even greater for the new problems now covered.

The algorithm has been implemented in {\sc Maple} giving promising experimental results.  TTICADs in general have much fewer cells than sign-invariant CADs using the same implementation and we showed that this allows even a simple implementation of TTICAD to compete with the state of the art CAD implementations.  For many problems the TTICAD theory offers the smallest truth-invariant CAD for a parent formula, and there are also classes of problems for which TTICAD is exactly the desired structure.  The benefits of TTICAD increase with the number of QFFs in a problem and is magnified if there is a QFF with no EC (as then the previous theory is not applicable).

\subsection{Future Work}
\label{subsec:FutureWork}

There is scope for optimizing the algorithm and extending it to allow less restrictive input.  Lemma \ref{lem:ConstPolys} gives one extension that is included in our implementation while other possibilities include removing some of the caution implied by well-orientedness, analogous to \citep{Brown2005a}.
Of course, the implementation of TTICAD used here could be optimised in many ways, but more desirable would be for TTICAD to be incorporated into existing state of the art CAD implementations.
In fact, since the ISSAC 2013 publication \cite{BCDEMW14} have presented an algorithm to build TTICADs using the \texttt{RegularChains} technology in {\sc Maple} and work continues in dealing with issues of problem formulation for this approach \citep{EBCDMW14, EBDW14}.

We see several possibilities for the theoretical development of TTICAD:
\begin{itemize}
\item Can we apply the theory recursively instead of only at the top level to make use of bi-equational constraints?  For example by widening the projection operator to allow enough information to conclude order-invariance, as in \citep{McCallum2001}.  

When doing this we may also consider further improvements to the lifting phase as recently discussed in \cite{EBD15}.
\item Can we make use of the ideas behind partial CAD to avoid unnecessary lifting once the truth value of a QFF on a cell is determined?
\item Can we implement the lifting algorithm in parallel?
\item Can we modify the lifting algorithm to only return those cells required for the application?  Approaches which restrict the output to cells of a certain dimension, or cells on a certain variety, are given by \cite{WBDE14}.  
\item Can anything be done when the input is not well oriented?
\end{itemize}


\begin{ack}
We are grateful to A.~Strzebo\'nski for assistance in performing the Mathematica tests and to the anonymous referees of both this and our ISSAC 2013 paper for their useful comments.  
We also thank the rest of the Triangular Sets seminar at Bath (A.~Locatelli, G.~Sankaran and N.~Vorobjov) for their input, and the team at Western University (C.~Chen, M.~Moreno Maza, R.~Xiao and Y.~Xie) for access to their {\sc Maple} code and helpful discussions.
\end{ack}

\bibliographystyle{elsarticle-harv}
\bibliography{CAD}




\end{document}